\documentclass[journal,onecolumn,draftclsnofoot]{IEEEtran}
\usepackage[final]{graphicx}
\usepackage[reqno]{amsmath}
\usepackage{amssymb}
\usepackage{cite}
\usepackage{balance}
\usepackage{epsfig}
\usepackage{color}
\usepackage{psfrag}
\usepackage{algpseudocode}
\usepackage{algorithm}

\newfont{\bbb}{msbm10 scaled 500}

\newfont{\bb}{msbm10 scaled 1100}

\newcommand{\Gc}{{\cal G}}

\newcommand{\Nc}{{\cal N}}

\newcommand{\Pc}{{\cal P}}

\newcommand{\Rc}{{\cal R}}
\newcommand{\Sc}{{\cal S}}
\newcommand{\Tc}{{\cal T}}

\newcommand{\argmax}{\operatornamewithlimits{arg\,max}}

\newcommand{\Ei}{\operatorname{Ei}}
\newcommand{\E}{\operatorname{E}}

\newtheorem{theorem}{Theorem}
\newtheorem{proposition}{Proposition}

\newtheorem{lemma}{Lemma}
\newtheorem{definition}{Definition}
\newtheorem{remark}{Remark}
\IEEEoverridecommandlockouts

\author{Karim Khalil and Eylem Ekici}

\title{Spectrum Access through Threats in Cognitive Radio Networks
\thanks{
The authors are with Department of Electrical and Computer Engineering, The Ohio State University, Columbus, OH 43210 USA. Email: \{khalilk,ekici\}@ece.osu.edu.}
}

\begin{document}
\maketitle

\begin{abstract}
We consider multiple access games in which primary users are interested in maximizing their confidential data rate at the minimum possible transmission power and secondary users employ eavesdropping as a leverage to maximize their data rate to a common destination at minimum transmission energy. For the non-cooperative static game, Nash equilibria in pure and mixed strategies are derived and shown to be Pareto inefficient in general, when channel gains are common knowledge. For the two-player Stackelberg game where the primary user is the leader, it is shown that the secondary user is forced to play as the follower where the Stackelberg equilibrium dominates the Nash equilibrium, even if the eavesdropper channel is better than the primary channel. Here, the utility achieved by the Stackelberg game Pareto-dominates the achieved Nash utility. Moreover, we study the unknown eavesdropper channel case numerically where the primary user has only statistical knowledge about the channel gain. We compare the results to the first scenario and show that it is not always beneficial for the cognitive user to hide the actual eavesdropper channel gain. Finally, we extend the equilibrium analysis to a multiple SU game where the primary system selects a subset of the secondary users to transmit such that the performance of the primary users is maximized.
\end{abstract}

\begin{keywords}
Cognitive radio, game theory, Stackelberg games, physical layer security, multiple access
\end{keywords}

\section{Introduction}
Research in Cognitive Radio Networks (CRNs) is motivated by wireless spectrum scarcity with respect to the ever growing demand of spectrum resources \cite{Akyildiz:NeXt06}. Recent studies show that the spectrum is inefficiently utilized by current licensees \cite{FCC}. Cognitive Radio has been proposed as means of solving this problem. The literature on CRNs can be divided into two main categories. In the first, which is often called the commons model, unlicensed users sense and access the spectrum of licensed systems as long as interference caused to the primary system is capped. Here, a primary system remains oblivious to the activity of secondary systems. In the second category, which is called the property-rights (or spectrum-leasing) model, interaction takes place between both systems such that unlicensed users are granted spectrum access only when they {\it add value} to the licensed system. This can be done, for instance, by rewarding the primary systems with monetary rewards \cite{Jayaweera:Dynamic09} or by improving its performance \cite{Khalil:Optimal11}. We call this type of interaction between primary and secondary systems {\it positive interaction}. In this paper, however, we study a different type of interaction (i.e., negative interaction) where secondary users threaten to compromise the privacy of primary systems.

Due to its broadcast nature, a big challenge in wireless networks is data privacy. Specifically, transmissions over the wireless channel are susceptible to eavesdropping. This issue is more important in cognitive radio settings, where unlicensed users may be equipped with advanced transceivers and are capable of communicating on multiple bands. Recently, the notion of physical layer secrecy was introduced which captures the wireless security from an information theoretic point of view \cite{Liu:Securing09}. Unlike conventional cryptography at higher layers, information theoretic secrecy allows for the development of provable and quantifiable secrecy measures without imposing restrictions on the computational capability of the eavesdroppers. These schemes are implemented based on signal processing and coding techniques at the physical layer. In this paper, we employ physical layer secrecy as a measure of privacy for the primary systems.

Traditionally, adversarial activity is adopted to cause damage to the attacked system by minimizing its utility. For example, malicious activity in a CRN was previously studied in \cite{Liu:Cognitive11}, where the objective of the malicious users is to degrade the performance of the primary system using different attacks like routing disruption and traffic injection. In \cite{Wu:Information11}, a Stackelberg game is studied in which transmissions of trusted cognitive users can improve the secrecy rates of primary users (PUs) with respect to an external eavesdropper. Other examples of security attacks in CRNs can be found in \cite{Liu:Cognitive11} as well.
In contrast to the existing work, in this paper, we consider adversarial activity where the cognitive users have a {\it different goal}. Specifically, in a system where both PUs and SUs\footnote{We use the terms 'user' and 'player'; as well as 'secondary' and 'cognitive' interchangeably throughout the paper.} transmit to a common destination (e.g., base station or access point), SUs employ eavesdropping as a leverage to gain access to the spectrum and maximize their own performance. This threat possibly forces PUs to lower their transmission power levels to improve their utility. Consequently, SUs can achieve higher rates since interference from the primary transmission to the secondary signal at the destination will be lowered, as well. Since the objectives of the different users are conflicting, we employ a game theoretical framework to study this adversarial situation. In our game formulations, we characterize equilibrium points and hence the optimal resource allocation for each of the primary and secondary systems. 

Game theory has been extensively employed in the analysis of wireless network problems in general and CRNs in particular \cite{Srivastava:Using05, Wang:Game10}. Specifically, the Multiple Access Channel (MAC) is one of the basic channel models that has been well studied using game theoretical techniques. For example, in \cite{Belmega:Power09}, a multiple access game is considered where a coordination signal is used to determine the order of the successive interference cancellation at the destination. Design of resource allocation algorithms for fading multiple access channels is studied in \cite{Lai:Water08} using both Nash and Stackelberg equilibrium concepts. In \cite{Karamchandani:Cooperation11}, cooperative random access and cooperative token ring are studied using coalitional game theory. A comprehensive survey on game theoretic approaches for interference free multiple access networks is presented in \cite{Akkarajitsakul:Game11}.

To the best of the authors knowledge, this work is the first to consider leveraging eavesdropping capabilities to access channel in a cognitive radio setting. We assume battery powered users and thus users are interested in saving energy. In our model, PUs are interested in maximizing their secrecy rate minus the cost of transmission power to a destination $D$. On the other hand, SUs wish to transmit best effort traffic to the common destination $D$ while minimizing their energy consumption. To this end, SUs threaten to eavesdrop the transmission of PUs to increase their own utility. In particular, in our model, SUs are equipped with half duplex wireless transceivers and employ a time division scheme where the available time is divided between transmitting own information and eavesdropping the information transmitted by PUs. All players are rational and selfish\footnote{In this paper, we study a different model for CRNs where no regulation is in place. This can model, for instance, the case when SUs has no connection to the internet and can not register to a centralized database.} and choose their strategy to maximize their own utility functions. The {\it contributions} of the paper are as follows. First, we study the two-player game and characterize the Nash Equilibrium (NE) in both pure and mixed strategies. For certain ranges of the channel coefficients and energy cost parameters, the equilibrium point is shown to be inefficient and results in a lose-lose situation. A leader-follower game is then formulated in which the PU is the leader who specifies its strategy and then the follower reacts so that its utility is maximized. In this case, Stackelberg Equilibrium (SE) is characterized and shown to dominate, in the sense of Simaan \cite{Simaan:Equilibrium77},  the NE and hence the follower is forced to comply with this strategy. In these games, we first assume that all channel gains are common knowledge to both players. Then, we analyze the more realistic scenario when PU has only statistical knowledge about the channel gain of the eavesdropper. Here, the cognitive user has the ability to either hide or reveal the actual value of the eavesdropper channel coefficient. Interestingly, we show that its not always beneficial for SU to hide this information from PU, especially when secondary channel gains are low. Next, we extend the analysis to multi-player games and characterize the optimal decoding order and equilibria for a special case of channel condition and cost parameters. Finally, we present an algorithm that is implemented by the primary system to select secondary users to transmit such that the utility of the primary system is maximized.

The rest of the paper is organized as follows. Section \ref{sec:background} provides the required background from game theory and information theory. In Section \ref{sec:nash}, we present our game setup and derive its NE for all values of channel conditions and energy cost parameters. In addition, we show via examples how the derived equilibrium points can be inefficient. Stackelberg formulation is then considered in Section \ref{sec:stackelberg} and the SE is shown to Pareto-dominate the NE for all ranges of channel coefficients when PU is the leader. The effect of unknown eavesdropper channel at PU is then considered in Section \ref{sec:unknownb}. In Section \ref{sec:multi}, we extend the game to multiple players. Finally, the paper is concluded in Section \ref{sec:conclusion}.

\section{Background}
\label{sec:background}
In this section, we review results from information theory about the multiple access channel and the wiretap channel that we employ in our game formulation. Moreover, we present certain definitions from non-cooperative game theory that are essential in our analysis.

\subsection{Multiple Access Channel and Wire-tap Channel}
\label{sec:mac}


The two-user multiple access channel is a well-known channel model in the network information theory \cite{Cover:Elements91}. Let the channel capacity function be defined as $C(x) = \frac{1}{2} \log(1+x)$, where logarithms are taken to the base $2$. The capacity region of a channel defines the achievable rates so that the receiver can decode the information reliably, i.e., with an arbitrarily small probability of decoding error. For the two-user Additive White Gaussian Noise (AWGN) Multiple Access Channel (MAC), the capacity region is a pentagon given by
\begin{align}
R_0(P_0) &\leq C(a P_0),\notag\\ R_1(P_1) &\leq C(c P_1),\notag\\
R_0(P_0) + R_1(P_1) &\leq C(a P_0 + b P_1),\label{eq:rate}
\end{align}
where $R_0(\cdot), R_1(\cdot)$ are the achievable rates for transmitters $0$ and $1$, respectively, $P_0, P_1$ are the transmission power levels, $a>0$, $c>0$ are the (constant) channel power gains and noise is assumed to have unit variance. 
We assign user $0$ to be the PU and user $1$ to be the SU. The two corner points of the capacity region are achieved by successive interference cancellation at the decoder where the order of decoding determines the corner point \cite{Cover:Elements91}. 

In this paper, we employ rate expressions on the boundary of the region \eqref{eq:rate}. Specifically, we assume that the destination 
always decodes SU first in the interference cancellation decoder and hence gives priority to PU's signal. Consequently, the achievable rates at the destination for PU and SU are given by
\begin{align}
R_0(P_0) &= C(a P_0),\notag\\
R_1(P_1) &= C(\frac{c P_1}{1+ a P_0}).\label{eq:r2}
\end{align}

We note that the rate pairs on the boundary of the capacity region \eqref{eq:r2} can only be achieved when both transmitters coordinate the codebooks and rates used in the channel coding \cite{Cover:Elements91}. This coordination can be facilitated by the base station. If no coordination is assumed, then interference will affect achieved rates of both users.

\begin{figure}[bht]
    \centerline{
        \scalebox{0.7}{
          \input{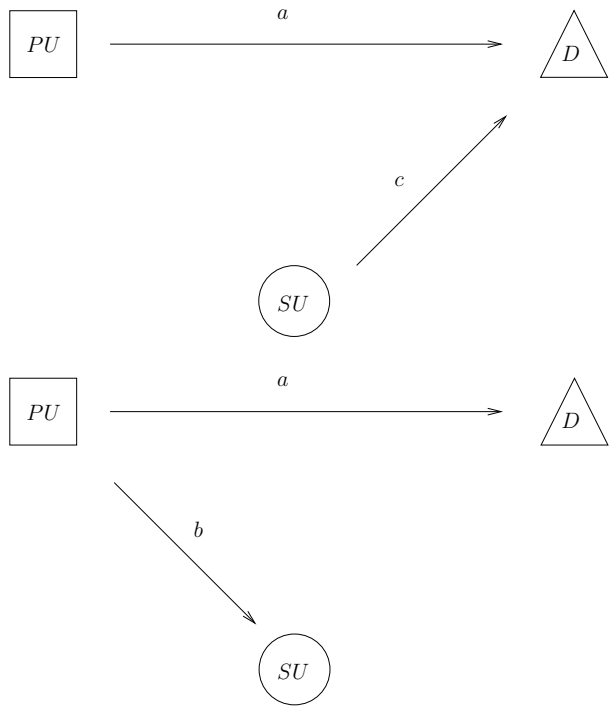}
         }
    }
    \caption{Channel model.}
    \label{fig:model2}
\end{figure}

In the presence of an eavesdropper, the achievable secrecy rate of a transmitter is the rate at which the entropy of the sender's message at the eavesdropper is arbitrarily close to the entropy of the message itself, given the received signals. In other words, the secure rate is the rate at which the message of the sender is almost independent from the received signals at the eavesdropper. Achievability schemes (i.e., channel coding) are designed to maximize the confusion at the eavesdropper while maximizing the reliable rate at the legitimate receiver by exploiting the wireless channel characteristics such as noise and fading. For the Gaussian channel, the secrecy capacity is given by \cite{Csiszar:Broadcast78}
\begin{align}
R_s (P) = \left[C(aP)-C(bP)\right]^+,
\end{align}
where $a,b>0$ are the channel gains of the legitimate receiver's channel and the eavesdropper channel, respectively. The wiretap channel is sketched in the lower part of Figure \ref{fig:model2}

In our game formulation, we assume that the SU employs a half duplex transceiver and can either transmit to $D$ or eavesdrop the transmission of PU at any given time. Thus, the channel model during SU's transmission is a multiple access channel, while it is a wiretap channel during eavesdropping, as shown in Figure \ref{fig:model2}. Throughout the paper, we refer to the channel between PU and $D$ as the primary channel, the channel between SU and D as the secondary channel, and the channel between PU and SU as the eavesdropper channel. 

Our network model is appropriate for CRNs for multiple reasons. First, when SU is transmitting information to the common destination $D$, it is given a lower priority than PU. This is clear from the achievable rate \eqref{eq:r2} where SU's  achievable rate decreases with increasing transmission power of PU while the opposite is not true. Through the threat of SU, PU may be forced to decrease its transmission power $P_0$ and hence SU achieves higher data rate, as will be discussed in the next section. Finally, as discussed in Section \ref{sec:stackelberg}, the analysis reveals that SU is forced to follow PU in a leader-follower game.

\subsection{Game Theory Basics}
\label{sec:gamebasics}
Game theory provides an analytical framework to analyze situations of conflict between multiple decision makers that are {\itshape{\bfseries rational, intelligent and selfish}}. These attributes accurately characterize wireless devices designed to optimize their own performance. Here, we borrow definitions from \cite{Basar:Dynamic95} and \cite{Simaan:Equilibrium77} that are needed for the equilibrium analysis in the following sections. A strategic game is any $\Gc$ of the form $\Gc=(\Nc,(\Sc_i)_{i\in\Nc},(u_i)_{i\in\Nc})$, where $\Nc$ is the set of players in the game. Let the utility of a player be given by $u_i(s_i,s_{-i})$ where $s_i\in\Sc_i$ is the pure (deterministic) strategy (or action) of player $i$, chosen from the set of available strategies $\Sc_i$ and $s_{-i}$ is the strategy profile of all other players, except for player $i$ chosen from $\times_{j\in\Nc-\{i\}} \Sc_j$. In the following definitions, we focus on two-player games, i.e., $\Nc=\{1,2\}$.
\begin{definition}
An NE point is a strategy pair $(s_1^*,s_2^*)$ such that
\begin{align}
u_0(s_0^*,s_1^*)\geq u_0(s_0,s_1^*), \forall s_0\in\Sc_0,\notag\\
u_1(s_0^*,s_1^*)\geq u_1(s_0^*,s_1), \forall s_1\in\Sc_1.
\end{align}
\label{def:nash}
\end{definition}
This definition implies that at an NE, no user has incentive to unilaterally deviate to other operating points. Assume there exist two well defined unique mappings $T_0:\Sc_1\rightarrow\Sc_0$ and $T_1:\Sc_0\rightarrow\Sc_1$ such that for any fixed $s_1\in\Sc_1$, $u_0(T0 (s_1), s_1)\geq u_0(s_0,s_1), \forall s_0\in\Sc_0$ and for any fixed $s_0\in\Sc_0$, $u_1(s_0, T_1 (s_0))\geq u_1(s_0,s_1), \forall s_1\in\Sc_1$, i.e., $T_i$ defines strategies that are best response to each strategy chosen by the other player. Let the set $D_i = \{(s_0,s_1)\in\Sc_0\times\Sc_1:s_i=T_i (s_j)\}$ for $i=0,j=1$ and $i=1,j=0$ be called the rational reaction set of player $i$ and let $D_i(s_j) = \{s_i\in\Sc_i:(s_i,s_j)\in D_i\}$. Note that any pair in the set $D_0\cap D_1$ is an NE according to Definition \ref{def:nash}. Hence, a strategy profile $S$ is an NE if and only if the strategy of every player in $S$ is a best response to the other player's strategy. 

When mixed strategies are allowed, Definition \ref{def:nash} is written for expected utilities instead. Let $f_i:\Sc_i\rightarrow [0,1]$ be a mixed strategy for player $i$ which defines a probability density function over $s_i$ such that $f_i(s_i)\geq 0 \forall s_i\in\Sc_i$. Here, we define $supp(f_i)$ as the set of actions for player $i$ with positive probability in the mixed strategy $f_i$, that is, $supp(f_i)=\{s\in \Sc_i\ :f_i(s)>0\}$. At any Mixed Strategy Nash Equilibrium (MSNE), strictly dominated actions can not be assigned positive probability \cite{Fudenberg:Game93}. Here, an action $s_i\in \Sc_i$ for player $i$ is strictly dominated if there exists a mixed strategy $f_i$ such that ${\mathbb E}[u_i(f_i,s_j)] > u_i(s_i,s_j)\ \forall s_j\in\Sc_j,\ j\neq i$. The following proposition helps in the characterization of MSNE \cite{Fudenberg:Game93}. 

\begin{proposition}
A mixed strategy profile $f^*=(f_0^*,f_1^*)$ is an MSNE if and only if for each player $i$ there exists some $c_i\in\Rc$ such that 
\begin{align}
{\mathbb E}[u_i(s_i,f_j)] =& c_i,\ \forall s_i\in supp(f_i),\\ {\mathbb E}[u_i(s_i,f_j)] <& c_i,\ \forall s_i\notin supp(f_i),\ i\neq j.\label{eq:property1}
\end{align}
\label{prop:prop}
\end{proposition}
That is, at an MSNE, each player chooses the support for its mixed strategy so that it contains only pure strategies that leads to best response to the other players strategies. One important result about MSNE is that for any continuous game, i.e., game with continuous utility functions, {\it there exists at least one NE} point in mixed strategies \cite{Fudenberg:Game93}.

A generalization of the strategic form game is the Bayesian game \cite{Myerson:Game91} where some players have private information about the game that other players do not have. A Bayesian game takes the form
\begin{align} \Gc_b=(\Nc,(\Sc_i)_{i\in\Nc},(\Tc_i)_{i\in\Nc},(p_i)_{i\in\Nc},(u_i)_{i\in\Nc})
\end{align}
where $\Tc_i$ is the set of {\itshape types} of the $i^{\text{th}}$ player that specifies the information player $i$ only knows about the game and $p_i$ is the probability function specifying what player $i$ believes about the other players' types given its own type. Here, it is assumed that each player $i$ knows the structure of $\Gc_b$ and its own type $t_i\in\Tc_i$. A {\itshape Bayesian Equilibrium} (BE) is an NE for $\Gc_b$ such that each player maximizes its expected utility. Therefore, the Bayesian game assumes that users are risk neutral.

The last type of game formulations we employ is Stackelberg games. In a Stackelberg game, the leader first makes a decision about its own strategy and then followers choose their strategies accordingly. In the following definitions, we fix player 1 as the leader and player 2 as the follower. The leader chooses the strategy that maximizes its utility from the rational reaction set of the follower.
\begin{definition}
A strategy $\bar{s}_0\in\Sc_0$ is a Stackelberg equilibrium strategy for the leader if
\begin{align}
\inf_{s_1\in D_1(\bar{s}_0)}  u_0(\bar{s}_0,s_1)\geq \inf_{s_1\in D_1(s_0)} u_0(s_0,s_1),\ \ \ \forall s_0\in \Sc_0.
\end{align}
\label{def:stackelberg}
\end{definition}

In this paper, we sometimes use the shorthand Stackelberg Equilibrium (SE) to mean Stackelberg equilibrium strategy. We also use the shorthand SEP (respectively SES) to indicate an SE with the primary user (respectively secondary user) as the leader.

The utility of the leader is a well defined quantity \cite{Basar:Dynamic95} and is given by 
\begin{align}
\bar{u}_0 = \sup_{s_0\in\Sc_0}\ \inf_{s_1\in D_1(s_0)} u_0(s_0,s_1).
\end{align}
A Stackelberg equilibrium strategy for the leader may not exist in general \cite{Basar:Dynamic95}. In this case, however, an $\epsilon$-SE can possibly exist in which the leader achieves utility $\epsilon$ close to $\bar{u}_0$.
\begin{definition}
Let $\epsilon > 0$ be a given real number. Then, a strategy $\bar{s}_{0\epsilon}\in\Sc_0$ is called an $\epsilon$-Stackelberg equilibrium strategy for the leader if 
\begin{align}
\inf_{s_1\in D_1(\bar{s}_{0\epsilon})} u_0(\bar{s}_{0\epsilon}, s_1) \geq \bar{u}_0 -\epsilon.
\end{align}
\label{def:ubar}
\end{definition}
One important property is that an $\epsilon$-SE always exists in a game $\Gc$ if $\bar{u}_0$ is finite \cite{Basar:Dynamic95}. 

From Definitions \ref{def:nash}, \ref{def:stackelberg}, it can be seen that the utility achieved by a user in a Stackelberg game under its own leadership is always at least as good as the utility achieved under any NE for the same game \cite{Basar:Dynamic95}. This fact motivates the following definition.
\begin{definition}
A Stackelberg equilibrium strategy $(\bar{s}_0,\bar{s}_1)$ with player 0 as the leader is said to dominate an NE $(s_0^*,s_1^*)$ if
\begin{align}
u_1(\bar{s}_0,\bar{s}_1)\geq u_1(s_0^*,s_1^*). \label{eq:dominance}
\end{align}
\label{def:dominance}
\end{definition}
In case \eqref{eq:dominance} is true, both the leader and the follower would better choose to play the Stackelberg game under the leadership of player 1. In Section \ref{sec:stackelberg}, this property will be vital to show that SU accepts to be a follower in a Stackelberg game that leads to more efficient performance for both players.


\section{Cognitive Eavesdropper Nash Game}
\label{sec:nash}
In this and the following sections, we consider a two-player static non-cooperative game $\Gc$ where the players are PU and SU. We consider more general case with multiple users in Section \ref{sec:multi}. In our games, we consider backlogged users where each user always has packets to transmit. We also consider battery powered users and thus saving evergy is important for all users. Here, the strategy of PU is to select the transmission power level $s_0 = P_0\in[0, P^{\text{max}}_0 ]$, while the strategy of SU is to choose the fraction $s_1=\alpha\in[0,1]$ by which it divides the total available time $T$ into transmission time $\alpha T$ and eavesdropping time $(1-\alpha)T$. Without loss of generality, we assume $T=1$. Note that PU is transmitting all the time while SU is either transmitting its own data or eavesdropping the primary traffic. 

PU is interested in maximizing its secrecy rate to the destination at the minimum power cost. The utility function of PU is given by
\begin{align}
u_0' (P_0,\alpha) = \left[C(a P_0)-(1-\alpha)C(b P_0)\right]^+ - J(P_0), \label{eq:u1'}
\end{align}
where $J(\cdot)$ is the power cost function which is increasing in $P_0$. In our analysis, we employ a linear power cost function $J(P_0) = \gamma P_0$, where $\gamma>0$ is the unit power cost. This choice is made only for analytical tractability of the analysis. However, the same equilibrium results also hold for general increasing cost functions $J(\cdot)$ such that $J(0)=0$ as discussed in Section \ref{sec:stackelberg}. The second term in \eqref{eq:u1} reflects the rate eavesdropped by the cognitive user. For example, when $\alpha = 1$, SU is transmitting all the time and the eavesdropped rate is zero. We note that we study a fundamental problem and thus we use the optimal multiple access scheme when both users are transmitting and the optimal secrecy coding scheme when only the PU is transmitting. In particular, there is no interference term due to the transmission of SU since we only consider the optimal multiple access scheme in which the traffic of SU is decoded first and then subtracted from the received signal using interference cancellation. Define the utility function
\begin{align}
u_0 (P_0,\alpha) = C(a P_0)-(1-\alpha)C(b P_0) - \gamma P_0. \label{eq:u1}
\end{align}
In our analysis, we use the form \eqref{eq:u1} for the utility function of PU, in place of \eqref{eq:u1'} to simplify the analysis. The following lemma shows that, for a given $\alpha$, the power level maximizing $u_0(\cdot)$, i.e. $P^*(\alpha)$, also maximizes $u_0'(\cdot)$.
\begin{lemma}
For a given $\alpha$, $P^*(\alpha)$ that maximizes $u_0(P_0,\alpha)$ also maximizes $u_0'(P_0,\alpha)$.
\label{lemma:approx}
\end{lemma}
\begin{proof}
Fix some $\alpha$. Suppose $u_0(\cdot)>0$ over some interval $I\subset(0,P_0^{\text{max}}]$ and $u_0(\cdot)\leq 0$ over $[0,P_0^{\text{max}}]-I$. Note that $u_0'(\cdot)=u_0(\cdot)$ over the interval $I$ and $u_0'(\cdot)\leq 0$ over $[0,P_0^{\text{max}}]-I$. Thus, $P^*(\alpha)\in I$ also maximizes $u_0'(\cdot)$. Now suppose that $u_0(\cdot)\leq 0$ over $[0,P_0^{\text{max}}]$. Thus, $u_0'(\cdot)\leq 0$ over $[0,P_0^{\text{max}}]$ and $P^* = 0$ also maximizes $u_0(\cdot)$.
\end{proof}
Note that the above proof is valid for any general energy cost function which is increasing in $P_0$.
We assume that SU is bounded by a maximum power constraint $ P^{\text{max}}_1 $. SU uses $ P^{\text{max}}_1 $ as its fixed transmission power level over the entire transmission period. In addition, we assume that SU is penalized for energy consumption using a linear cost function. The utility function of SU is thus given by
\begin{align}
u_1(P_0,\alpha) = \alpha \left(C\left(\frac{c  P^{\text{max}}_1 }{1+ a P_0}\right)-\beta,\right)
\label{eq:u2}
\end{align}
where $\beta>0$ is the energy cost per unit transmission time. In this section and in Section \ref{sec:stackelberg}, we assume that $a,b,c,\gamma,\beta, P^{\text{max}}_1 ,  P^{\text{max}}_0 $ are common knowledge. Hence, the games considered in these sections are non-cooperative static games with {\it complete information}. The goal of each user is to maximize its own utility by selecting the appropriate strategy given the knowledge of the other user's utility function. 

The following notation will be useful in our analysis of the game $\Gc$. We define $P^*(\alpha)$ as PU's power level that maximizes $u_0(P_0,\alpha)$ for a given $\alpha\in[0,1]$. Hence, $P^*(0)$ and $P^*(1)$ are PU's power levels that maximize the functions $u_0(P_0,0)$ and $u_0(P_0,1)$, respectively. Also we define the threshold power level $Q$ as 
\begin{align}
Q = \frac{1}{a}\left(\frac{c  P^{\text{max}}_1 }{2^{2 \beta}-1}-1,\right)
\label{eq:Q}
\end{align}
where the slope of the function $u_1(P_0,\alpha)$ is positive if $P_0 < Q$.

To ease our characterization of the NE of $\Gc$, the result is separated into two cases. In the first, the primary channel is assumed to be stronger than the eavesdropper channel, i.e., $a\geq b$. Here, we focus our analysis only on pure strategies. This is made possible without any loss of generality since the strategy sets are convex and the utility functions are concave in the corresponding variables \cite{Basar:Dynamic95} in this case, as will be shown in the proof of Theorem \ref{thm:nash1}. In the second case, the complementary case $a < b$ is considered. In this case, mixed strategy Nash equilibria may exist and hence an extended argument is required. The result of this scenario is presented in Theorem \ref{thm:nash2}. The following Lemma characterizes the structure of the function $u_0(P_0,\alpha)$.

\begin{lemma}
\label{lemma:u1}
For a given $\alpha$, $P^*(\alpha)$ maximizes the function $u_0(P_0,\alpha)$ with respect to $P_0$, where
\begin{align}
P^*(\alpha) &= \argmax_{P_0\in\{0,P'(\alpha)\}} u_0(P_0,\alpha),\label{eq:P*}\\
P'(\alpha) &= \argmax_{P_0\in [\hat{P}(\alpha), P^{\text{max}}_0 ]} u_0(P_0,\alpha) = \frac{X(\alpha)+\sqrt{X(\alpha)^2-Y(\alpha)}}{2\bar{\gamma}ab}, \label{eq:P'}
\end{align}
and
\begin{align}
\hat{P}(\alpha)&\triangleq\frac{b\sqrt{1-\alpha}-a}{ab(1-\sqrt{1-\alpha})}, \label{eq:concavitycond}\\
X(\alpha) &=  \alpha a b -\bar{\gamma} (a+b),\\
Y(\alpha) &= 4\bar{\gamma} a b (\bar{\gamma}-a+b(1-\alpha)).
\end{align}
\end{lemma}
\begin{IEEEproof}
Given some strategy $\alpha$ of SU, the best response for PU, i.e., $P^*(\alpha)$ is given as 
\begin{align}
P^*(\alpha) =\argmax_{P_0\in[0, P^{\text{max}}_0 ]} u_0(P_0,\alpha).\label{eq:argmaxP}
\end{align}
Using the second derivative test, it can be shown that $u_0(P_0,\alpha)$ is concave in $P_0$ if $P_0\geq\hat{P}(\alpha)$ and convex otherwise.
For the concave region (i.e., $P_0\in [\hat{P}(\alpha), P^{\text{max}}_0 ]$) and by setting the first derivative of $u_0(\cdot)$ with respect to $P_0$ to zero and solving for $P_0$, we get the expression for $P'(\alpha)$.

For the convex region, the maximum is on the boundary of the interval $[0,\hat{P}(\alpha)]$. Then, from the continuity of $u_0(\cdot)$, the expression for $P^*(\alpha)$ follows and the proof is complete.
\end{IEEEproof}

The following Theorem characterizes the unique NE of the game $\Gc$ for the case $a\geq b$.

\begin{theorem}
For the game $\Gc$, if $a\geq b$, then the unique NE is
\begin{align}
(s_1^*,s_2^*)=
\begin{cases}
\left(P'(0),0\right);\ &\text{if } Q< P'(0)\\
\left(Q,\alpha_Q\right);\ &\text{if } P'(0) \leq Q \leq P'(1)\\
\left(P'(1),1\right);\ &\text{if } P'(1)<Q.
\end{cases}
\label{eq:ne1}
\end{align}
where $\alpha_Q\in[0,1]$ is the time fraction of SU that solves the equation $P'(\alpha)=Q$ and $\bar{\gamma} = \gamma \ln(4)$.
\label{thm:nash1}
\end{theorem}
\begin{IEEEproof}
We show that the intersection of the best response correspondences of PU and SU are exactly the points in the Theorem. Therefore, no user has incentive to deviate unilaterally from such points and the conditions of Definition \ref{def:nash} are satisfied at these given points. 

From Lemma \ref{lemma:u1}, note that when $a\geq b$, $\hat{P}(\alpha)<0\ \forall \alpha \in [0,1]$ and hence $P_0\geq\hat{P}(\alpha)$ for all $P_0\in[0, P^{\text{max}}_0 ]$ and $P^*(\alpha)= P'(\alpha)$.

The utility of SU $u_1(P_0,\alpha)$ is linear in $\alpha$ given $P_0$ and the slope of $u_1(P_0,\alpha)$ is non-positive for $P_0\geq Q$. Consider the case $Q < P'(0)$ and assume NE is at $\alpha = 0$. Then, the best response for the primary player is $P_0 = P'(0)>Q$ implying that $(P'(0),0)$ is an NE in this case. This proves the first case.

Now assume $\alpha^* = 1$. Then, the best response for the primary is $P_0 = P'(1)$. When $P'(1) < Q$, the best response for SU when $P_0 =P'(1)$ is $\alpha = 1$ and hence $(P'(1),1)$ is an NE in this case.

For the remaining case, i.e., $P'(0)\leq Q\leq P'(0)$, we note that SU is indifferent to the choice of $\alpha$ when PU chooses $P_0=Q$ since the slope of $u_1(Q,\alpha)$ with respect to $\alpha$ is zero in this case. The intersection of the best response sets for the PU and SU is at $\alpha^*=\alpha_Q$. The solution $\alpha_Q$ to the equation $P'(\alpha)=Q$ is given by
\begin{align}
\alpha_Q = \frac{\bar{\gamma}\left[Q(a+b+abQ)+1\right]-a+b}{b(aQ+1)}
\end{align}
The parameter $\alpha_Q$ is in the interval $[0,1]$ if and only if $P'(0) \leq Q \leq P'(1)$ implying that it is the only NE in this case. Finally, note that $P'(\alpha)$ is an increasing function in $\alpha$. Therefore, the relation $P'(0)<P'(1)$ is always valid and we do not need to consider other cases. This concludes the proof.
\end{IEEEproof}

The unique NE for the game $\Gc$ suggests the following. SU's decision depends on the choice of the power level by PU compared to the threshold value $Q$. Only when $P_0<Q$, SU is able to achieve positive utility. 
If the threshold $Q$ is high so that $P^*(1)\leq Q$ (the last case of \eqref{eq:ne1}), then SU transmits all the time and the eavesdropping term in the primary utility function vanishes. When $\alpha = 1$, we have
\begin{align}
P^*(1) = \min\Big\{ P^{\text{max}}_1 ,\Big[\frac{1}{\bar{\gamma}} - \frac{1}{a}\Big]^+\Big\}.
\label{eq:P11}
\end{align}
An example of such a case is when the primary channel gain $a$ is small and the secondary channel gain $c$ is high for a given $\beta$. This can be seen from \eqref{eq:Q} and \eqref{eq:P11}. In this case, interference from PU to SU at the destination is low and SU chooses to transmit for the entire available time. In addition, PU is not affected by the transmission of SU.

On the other hand, when $c$ is small and $a$ is large, we might have $Q<P^*(1)$. In this case, SU achieves zero utility (the first ands second case of \eqref{eq:ne1}). In fact, SU is forced to select $\alpha=0$ in order to avoid negative utility in these cases. 

Now we present a similar result for the case when $a<b$. We start by computing the best response correspondences for each user. It is straight forward to see that for SU, $T_1(P_0) = 1$ if $P_0<Q$, $T_1(P_0) = 0$ if $P_0>Q$ and $T_1(P_0) = \alpha,\ \alpha\in[0,1]$ if $P_0=Q$, similar to the case when $a>b$. For PU, note that $P'(\alpha)$ is an increasing function of $\alpha$. To find $T_0(\alpha)$, the following definition is needed. Let $\tilde{\alpha}$ be defined such that $u_0(P'(\tilde{\alpha}),\tilde{\alpha})= 0$ and $u_0(P'(\alpha),\alpha)>0 \ \forall \alpha >\tilde{\alpha}$. From \eqref{eq:P*}, we have $T_0(\alpha)=0$ if $\alpha\in[0,\tilde{\alpha})$, $T_0(\alpha)=P'(\alpha)$ if $\alpha\in(\tilde{\alpha},1]$, and $T_0(\tilde{\alpha})=\{0,P'(\tilde{\alpha})\}$. Note that if $\tilde{\alpha}<0$ then $P'(\alpha)=P^*(\alpha)$ which is similar to the case when $a\geq b$. If $\tilde{\alpha}>1$, then $u_0(\cdot)\leq 0$ for all $P_0$ and $\alpha$, implying that the decision of PU is $P_0 = 0$.

The following theorem characterizes the Nash equilibria when $a<b$. In the following, $P_0(x)$ denotes the probability that PU uses the (discrete) action $x\in[0,P_0^{\text{max}}]$ in some mixed strategy.
\begin{theorem}
\label{thm:nash2}
For the game $\Gc$, if $a<b$, the following are the only Nash equilibrium points
\begin{align}
(s_0^*,s_1^*)=
\begin{cases}
(0,0);\ &\text{if } Q< 0\\
(0,g_1);\ &\text{if } Q = 0\\
\left(f_2,g_2\right);\ &\text{if }0< Q< P'(\tilde{\alpha})\\
(Q,g_3);\ &\text{if } P'(\tilde{\alpha}) \leq Q \leq P'(1)\\
\left(P'(1),1\right);\ &\text{if } P'(1)<Q.
\end{cases}
\label{eq:ne2}
\end{align}
where $g_1$ is the mixed strategy for SU with an arbitrary probability distribution over $supp(g_1)=[0,\tilde{\alpha}]$, $f_2$ is the mixed strategy for PU with $supp(f_2)=\{0,P'(\tilde{\alpha})\}$ and $P_0(0)$, $P_0(P'(\tilde{\alpha}))$ are the unique solutions to both linear equations
\begin{align}
P_0(0)+ P_0(P'(\tilde{\alpha}))=1,\notag\\
C(c  P^{\text{max}}_1 )P_0(0)  + C\left(\frac{c  P^{\text{max}}_1 }{1+ a P'(\tilde{\alpha})}\right)P_0(P'(\tilde{\alpha}))=\beta,
\label{eq:f}
\end{align}
such that $P_0(0)>0, \ P_0(P'(\tilde{\alpha}))>0$. In addition, $g_2,g_3$ are the mixed strategy for SU such that $supp(g_2)=supp(g_2)=[0,1]$ and ${\mathbb E}[\alpha]=\tilde{\alpha}, {\mathbb E}[\alpha]=\alpha_Q$ , respectively.
\end{theorem}
\begin{proof}
In the first case of \eqref{eq:ne2}, $u_1(P_0,\alpha)=0\ \forall \alpha>0$. Then the best response of SU is $\alpha^*=0\ \forall P_0$. Since $a<b$, then $u_0(P_0,0)<0\ \forall P_0>0$ implying that $P_0^*=0$. This proves the first case. 

For the second case of \eqref{eq:ne2}, the intersection of the best response correspondences of the two players is $(0,\alpha),\alpha\in [0,\tilde{\alpha}]$, and hence we have an infinite number of pure strategy Nash equilibria. For any MSNE $(f',g')$, the expected value of the slope of $u_1(\cdot)$ equals zero implying that $supp(f')=\{0\}$. Hence, $g'$ can have any distribution on $[0,\tilde{\alpha}]$. Note that all equilibria in this case have the property that ${\mathbb E}[u_0(\cdot)]={\mathbb E}[u_1(\cdot)]=0$.

For the third case, the intersection of the best response sets is empty and thus there is no pure strategy Nash equilibrium. To find the support of the mixed strategy of PU, note that $P_0\in\{(0,P'(\tilde{\alpha}))\}\cup \{P_0>P'(1)\}$ is never a best response to any strategy of SU and hence it can be discarded from the support. Using properties of MSNE from Section \ref{sec:gamebasics} and the fact that there exists at least one mixed strategy Nash equilibrium, then for some $\alpha'\in[0,1]$, there exists some mixed primary strategy $f'$ such that $u_0(P_0,\alpha')=K\ \forall P_0\in supp(f')$, where $K$ is some constant. Observing the best response set of PU, it can be seen that $\alpha'=\tilde{\alpha}$ is the only SU strategy satisfying this condition if the support of SU is singleton. In addition, $supp(f')=\{0,P'(\tilde{\alpha})\}$ and $K=0$. To find $f'$, SU must have no incentive to deviate from $\alpha=\tilde{\alpha}$, i.e., ${\mathbb E}[u_1(f',\tilde{\alpha})]=0$. It can also be seen that for any distribution $g'$ of $\alpha$ on the interval $[0,1]$ such that ${\mathbb E}[\alpha]=\tilde{\alpha}$, $(f',g')$ is an MSNE. This imply the third case in \eqref{eq:ne2}.

For the fourth case, the intersection of $T_0(\cdot)$ and $T_1(\cdot)$ implies the unique pure NE point $(Q,\alpha_Q)$. It can also be seen that if SU randomizes its action over the interval $[0,1]$ such that ${\mathbb E}[\alpha]=\alpha_Q$, then best response of PU is $P_0=Q$. Any other (mixed) strategy for SU will not result in a NE.

Finally, when $Q>P'(1)$, then $\alpha^*=1$ for all $P_0$ in the rational reaction set of PU. In addition, when $P_0 = P'(1)$, no player has incentive to deviate unilaterally. This concludes the proof.
\end{proof}

We now discuss interesting properties of the Nash equilibria of $\Gc$ that motivates our investigation of the Stackelberg game and lead to the main result of the paper. While the following discussion holds for both considered scenarios of channel conditions, we only focus on the scenario $a\geq b$ for brevity.

We observe that, in the cases where $Q\leq P'(1)$, the NE point is inefficient: there may be other operating points where at least one player achieves higher utility while the utility of the others is not decreased. For example, assume $0 < Q$. In the first case where $Q < P'(0)$, if PU chooses a power level less than but arbitrarily close to $Q$, then SU is willing to transmit all time, i.e., chooses $\alpha=1$. In this case, SU achieves a strictly positive utility rather than zero utility achieved at the NE. In addition, since the eavesdropping term in $u_0(\cdot)$ vanishes in this case, PU can achieve a better utility if $u_0(Q,1)>u_0(P'(0),0)$. The same argument is valid for the second case where $P'(0) \leq Q \leq P'(1)$.

To elaborate on our observation, consider the following numerical example. Let $a=2.5, b=1, c=3.5,  P^{\text{max}}_1  =  P^{\text{max}}_0  = 1, \beta=1, \bar{\gamma}=1$. According to \eqref{eq:Q} and \eqref{eq:P'}, these values imply that $Q < P'(0)$. In this case, the utility of PU at the NE is $u_0(P'(0),0) = 0.0211$ while its utility at $P_0 =Q$ and $\alpha=1$ is $u_0(Q,1) = 0.0631$ which is three times better. When $c$ is changed to $c = 5$ while fixing the rest of parameters, we have $P'(0) \leq Q \leq P'(1)$. Here, $\alpha_Q = 0.3667$, $u_0(Q,\alpha_Q)=0.0681$ and $u_0(Q,1)=0.1761$, which is more than two times  better than the primary utility achieved at the NE point. The utility of PU is sketched in Figure \ref{fig:numerical} for the cases $\alpha=0$, $\alpha=\alpha_Q$ and $\alpha=1$ when $c=5$. Each curve is parameterized by a value for $\alpha$. The points maximizing each case are marked in addition to the the suggested operating point $u_0(Q,1)$.
\begin{figure}[thb]
    \centering
    \psfrag{x1}[c]{\tiny $u_0(P'(0),0)$}
    \psfrag{x2}[c]{\tiny $u_0(Q,\alpha_Q)$}
    \psfrag{x3}[c]{\tiny $u_0(Q,1)$}
    \psfrag{x4}[c]{\tiny $u_0(P'(1),1)$}
    \includegraphics[width=0.75\columnwidth]{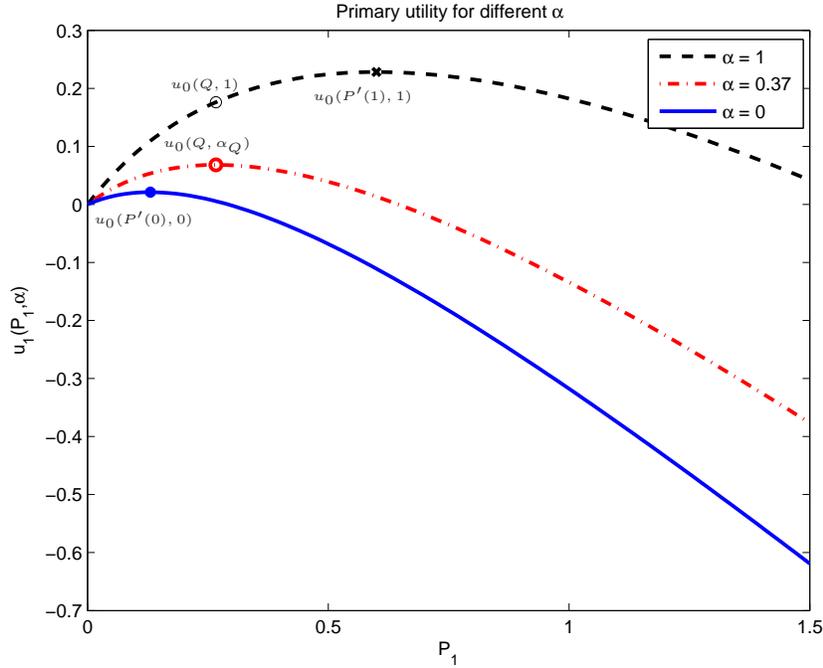}
    \caption{Primary utility for the case $P'(0) \leq Q \leq P'(1)$. Inefficient NE marked on the middle curve and suggested operating point circled on the upper curve.}
    \label{fig:numerical}
\end{figure}

It is important to note that the aforementioned operating points are not equilibrium points of the non-cooperative strategic game. For instance, in the case $Q < P'(0)$, if SU chooses $\alpha=1$, PU can take advantage of this choice and select $P_0=P'(1)$ which will cause SU to achieve negative utility. Therefore, without communication and contracts between players, there is no guarantee that both players will play the strategy profile $(Q-\epsilon,1)$ for an arbitrarily small $\epsilon>0$. However, if both players agree to play the game with some order and not take decisions simultaneously, better equilibrium points can be reached. In the following section, we formulate a leader-follower game in which the inefficient NE points of $\Gc$ are alleviated.

\section{Cognitive Eavesdropper Stackelberg game}
\label{sec:stackelberg}
In this section, we show the existence of a Stackelberg strategy under the leadership of PU that results in better payoff values for both players in $\Gc$ compared to the payoffs achieved at the NE. In a Nash game, each players choose their strategies independently of the actual choice of other players. This property of a strategic game can be viewed as if players decide their choices simultaneously. Alternatively, it can be viewed as a sequential decision process where each player has no information about the decisions of other players. 

On the other hand, in a Stackelberg game, the leader of the game chooses a strategy first and then announces it to the other players in the game (followers). Then the followers react to the strategy of the leader to maximize their own utilities. This leader-follower scenario can model the situation when a player has the power to enforce other players to be followers. In addition, a rational player is willing to play the Stackelberg game as a follower if this implies a better utility than that achieved at the NE of the game. If all players in the game achieve higher utility in the Stackelberg game with the leadership of some player compared to that achieved at the NE of the game, then this SE is said to dominate the NE according to Definition \ref{def:dominance}. 

For our game $\Gc$, we have two possible Stackelberg games: in the first, PU is the leader while in the second, SU is the leader. Here, we show that an SE with PU as the leader (denoted by SEP) dominates the NE in Section \ref{sec:nash}. Moreover, we show that any SE with SU as the leader (denoted by SES) can not dominate the NE of $\Gc$. This implies that SU is willing to be a follower in a Stackelberg game in order to achieve better utility values.

To check the existence of an SE for our game, we start by computing $D_2$, the rational reaction set for SU. It can be seen that
\begin{align}
D_1=\{(P_0,0):P_0 > Q\} \cup \{(P_0,1) : P_0 < Q\}\notag\cup \{(Q,\alpha) :\alpha\in[0,1]\}.
\label{eq:D2}
\end{align}
The following two lemmas establish the existence of an SEP and prove its dominance with respect to the NE for the game $\Gc$, for the two different cases of channel conditions considered in Section \ref{sec:nash}
\begin{lemma}
For the game $\Gc$ with $a\geq b$ and for any given $\epsilon>0$, if the channel gains $a,b,c$ and cost parameters $\gamma,\beta$ and $ P^{\text{max}}_1 , P^{\text{max}}_1 $ are finite, then there exists an $\epsilon$-SEP. Moreover, if $\epsilon$ is sufficiently small, then the $\epsilon$-SEP dominates the NE of $\Gc$.
\label{lemma:sep1}
\end{lemma}
\begin{IEEEproof}
For the existence part, it suffices to show finiteness of $\bar{u}_0$ as in \cite{Basar:Dynamic95}.  From Definition \ref{def:ubar} and from \eqref{eq:D2}, the utility of the PU for a Stackelberg game $\Gc$ under its leadership can be calculated as
\begin{align}
\bar{u}_1 =
\begin{cases}
u_0(P'(0),0); &\text{if } Q \leq 0\\
\max\{u_0(Q,1),u_0(P'(0),0)\}; &\text{if }0 < Q< P'(0)\\
u_0(Q,1); &\text{if } P'(0) \leq Q \leq P'(1)\\
u_0(P'(1),1); &\text{if } P'(1)<Q.
\end{cases}
\end{align}
If $a,b,c$ and $\gamma, \beta,  P^{\text{max}}_0 ,  P^{\text{max}}_1 $ are finite, then $Q$ and $P'(\alpha)$ are finite for all $\alpha\in[0,1]$. Then, $\bar{u}_0$ is finite in all cases and existence follows from \cite[Property 4.2]{Basar:Dynamic95}.
Now, fix some $\epsilon>0$ and consider the following strategy for PU.
\begin{align}
s_{0\epsilon}^1 =
\begin{cases}
P'(0); &\text{if } Q \leq 0\\
\underset{P_0\in\{Q-\epsilon,P'(0)\}}{\operatorname{arg\,max}} u_0(P_0,D_1(P_0)); &\text{if }0<Q < P'(0)\\
Q-\epsilon; &\text{if } P'(0) \leq Q \leq P'(1)\\
P'(1); &\text{if } P'(1) < Q.
\end{cases}
\label{eq:s1sep1}
\end{align}
For the first and last cases in \eqref{eq:s1sep1}, it can be seen that $u_0(s^1_{0\epsilon},D_1(s^1_{0\epsilon}))=\bar{u}_1$. In addition, for the other two cases, since $u_0(\cdot)$ is uniformly continuous in $P_0$, it can be seen that we have $u_0(s^1_{0\epsilon},D_1(s^1_{0\epsilon}))\geq\bar{u}_0 -\epsilon$ which holds for all cases of \eqref{eq:s1sep1}. This implies that $s^1_{0\epsilon}$ is in fact an $\epsilon$-SEP by definition. Finally, to show dominance of the above SEP over NE of $\Gc$, note that $\bar{u}_0\geq u_0(s_0^*,s_1^*)$. Then, for $\epsilon$ sufficiently small, $u_0(s^1_{0\epsilon},D_1(s^1_{0\epsilon}))$ is sufficiently close to $\bar{u}_0$. For the utility of SU, it can be seen that $u_1(\cdot)$ is the same as in the NE \eqref{eq:ne1} for the first and last cases of \eqref{eq:s1sep1} and $u_1(s^1_{0\epsilon},D_1(s^1_{0\epsilon}))\geq u_1(s_0^*,s_1^*)$ concluding the proof.
\end{IEEEproof}
The proof of the following lemma follows a similar argument to the proof of Lemma \ref{lemma:sep1} and is omitted for brevity.
\begin{lemma}
\label{lemma:sep2}
For the game $\Gc$ with $a<b$ and for any $\epsilon>0$, the following strategy for the PU is an $\epsilon$-SEP.
\begin{align}
s_{0\epsilon}^2 =
\begin{cases}
0; &\text{if } Q \leq 0\\
Q-\epsilon; &\text{if }0<Q \leq P'(1)\\
P'(1); &\text{if } P'(1) < Q.
\end{cases}
\label{eq:s1sep2}
\end{align}
Moreover, $s_{0\epsilon}^2$ dominates the NE of $\Gc$ for small enough $\epsilon$.
\end{lemma}

By Definition \ref{def:dominance}, since the SEP in Lemma \ref{lemma:sep1} and Lemma \ref{lemma:sep2} dominates NE of $\Gc$, then the SU will prefer to be the follower in a Stackelberg game under leadership of the primary than to play Nash. At the NE of $\Gc$, SU achieves zero utility for $Q\leq  P'(1)$. However, at the SEP, SU achieves a strictly positive utility value in the third case in \eqref{eq:s1sep1} and non-negative utility (according to channel conditions) in the second case.

Nevertheless, SU may prefer to play a Stackelberg game under its own leadership and not to follow PU to achieve better utility. The following Lemma, however, shows that for the game $\Gc$, no SES dominates the NE. This result shows that {\it PU can in fact enforce SU to be a follower in a Stackelberg game.}

\begin{lemma}
For the game $\Gc$, there exists no SE under the leadership of SU that dominates the NE.
\label{lemma:ses}
\end{lemma}
\begin{IEEEproof}
Consider the scenario $a\geq b$. We start by computing the rational reaction set $D_0$. It can be easily seen that 
\begin{align}
D_0=\{(P_0,\alpha):P_0=P'(\alpha), \alpha\in[0,1]\}.
\end{align}
Now we check the SES and compare it to the NE point of $\Gc$. The SES is given by
\begin{align}
\bar{\alpha} &= \argmax_{(s_0,s_1)\in D_0} u_1(s_0,s_1)\notag\\
&= \argmax_{\alpha\in[0,1]} \alpha\ (C(\frac{c  P^{\text{max}}_1 }{1+aP'(\alpha)})-\beta). \label{eq:alphases}
\end{align}
We start by comparing to the last case in \eqref{eq:ne1}. Suppose the maximizer of \eqref{eq:alphases} is $\bar{s}_1=\bar{\alpha}=1$. Then, $\bar{s_0}=P'(1)$ and we have $u_0(s_0^*,s_1^*)=u_0(\bar{s}_0,\bar{s}_1)$. If $\bar{\alpha}<1$ and since $u_0(\cdot)$ can only decrease by decreasing $\alpha$, then $u_0(s_0^*,s_1^*)>u_0(\bar{s}_0,\bar{s}_1)$ and the SES is not dominant in this case. For the first case in \eqref{eq:ne1}, it is easy to see that $u_0(s_0^*,s_1^*)=u_0(\bar{s}_0,\bar{s}_1)$. Finally, for the middle case, it can be seen that $\bar{\alpha} \leq \alpha_Q$ implying that $u_0(s_0^*,s_1^*)\geq u_0(\bar{s}_0,\bar{s}_1)$. 

Now consider the scenario $a<b$. Here, $D_0=\{(0,\alpha):\alpha\in[0,\tilde{\alpha}]\}\cup\{(P^*(\alpha),\alpha):\alpha\in[\tilde{\alpha},1]\}$. It can be seen that $u_1(\alpha,D_0(\alpha))$ is linear and increasing for $\alpha < \tilde{\alpha}$ and is multivalued at $\alpha=\tilde{\alpha}$. In addition, it is monotonically decreasing for $\alpha>\tilde{\alpha}$, with a discontinuity at $\alpha=\tilde{\alpha}$ such that $u_1(\tilde{\alpha}^-,D_0(\tilde{\alpha}^-))>u_1(\tilde{\alpha}^+,D_0(\tilde{\alpha}^+))$. Hence, there exists an $\epsilon$-SES at $\bar{s}_1=\bar{\alpha}=\tilde{\alpha}-\epsilon$ implying that $u_0(\bar{s}_0,\bar{s}_1)=0<u_0(s^*_0,s^*_1)$ at such equilibrium. This concludes the proof.
\end{IEEEproof}

As given in Lemma \ref{lemma:ses}, at any SES  of $\Gc$ and comparing to the first two cases in \eqref{eq:ne1}, SU can choose $\alpha$ that leads to a larger secondary utility. However, this choice can only degrade the primary utility $u_0(\cdot)$ and hence no SES dominates NE of $\Gc$ according to Definition \ref{def:dominance}. The results of this section are summarized in the following theorem where the proof follows from Lemmas \ref{lemma:sep1}, \ref{lemma:sep2} and \ref{lemma:ses} and the fact that PU can threaten SU to play the NE \eqref{eq:ne1}, \eqref{eq:ne2}.
\begin{theorem}
For the game $\Gc$, SU accepts to play as the follower and the outcome of the game is the SEP point in Lemma \ref{lemma:sep1} and Lemma \ref{lemma:sep2}.
\label{thm:sep}
\end{theorem}

\begin{remark}
Similar results can be derived for other threat games. For example, a threat game is recently considered in \cite{Khalil:Multiple12} where SU is a jammer that divides available transmission time between transmitting own information and transmitting noise symbols.
\end{remark}
\begin{remark}
Since $u_0(\cdot)$ is decreasing in $\alpha$ and $P^*(\alpha)$ is increasing in $\alpha$ for an increasing cost function $J(P_0)$ with $J(0)=0$, then the main results in the paper will hold for the general utility function $u_0'(P_0,\alpha)$. In particular, the structure of the NE derived in Theorems \ref{thm:nash1} and \ref{thm:nash2} is preserved given some $P^*(0),P^*(1)$. Then, the argument in Theorem \ref{thm:sep} will also be true. This fact can be illustrated using Figure \ref{fig:numerical}. For the general form of $u_0(\cdot)$, if $P^*(0)\leq Q\leq P^*(1)$, PU will choose $P_0 = Q-\epsilon$ for some $\epsilon>0$ so that it can operate on the curve $u_0(P_0,1)$ rather than the curve $u_0(P_0,Q)$ when $P^*(0)\leq Q\leq P^*(1)$ and $u_0(P_0,0)$ when $Q< P^*(0)$.
\end{remark}

Given that both players are rational and that both consider Nash and Stackelberg games, it is clear from Theorem \ref{thm:sep} that both players will choose to play the Stackelberg game with PU as leader and SU as follower in all cases of channel conditions and energy cost parameters. 

It can be seen from \eqref{eq:u2} that when the secondary channel is weak, i.e., $c$ is small, SU achieves negative utility if it is transmitting all the time. By threatening the PU via eavesdropping, SU forces PU to play a Stackelberg game, which enables the SU to achieve a strictly positive utility. However, it is interesting to note that since PU is the leader in this game, it specifies how much transmission is allowed to the SU by choosing $\epsilon>0$. No matter how small $\epsilon$ is chosen, SU is forced to comply with this specification. When the secondary channel is strong, i.e., $c$ is large, interference on SU from PU is negligible, and therefore SU transmits all the time. In addition, PU achieves the largest possible utility in this case.

The following is an example practical scenario for the game we consider: an SU joins an existing primary system and measures the received signals from PU to estimate the eavesdropper channel gain $b$. Assume that maximum power levels and energy cost are common knowledge in the game. Since $D$ estimates the primary channel gain $a$, SU can also decode the feedback signals from $D$ to PU to know $a$. Then, SU announces its presence and therefore $D$ can measure the received signals from SU, estimate $c$ and then feedback this value to PU. Next, we focus on the general case of asymmetrical channels where PU can not estimate $b$ from the received signals from SU on the reverse channel. To this end, SU has information about $a,b$ and $c$ while PU has information about $a$ and $c$ only. PU can play the Stackelberg game only when it has information about $b$, in which case it announces its strategy, i.e., the primary power level and the value of $\epsilon$. However, SU may not want to share the private information about the actual value of $b$ with PU in order to improve its own utility. This motivates our analysis in the next section.

\section{Hidden Eavesdropper Channel}
\label{sec:unknownb}
In this section, we relax the assumption of the knowledge of the eavesdropper channel gain $b$ at the PU that was used in the analysis in sections \ref{sec:nash} and \ref{sec:stackelberg}. Specifically, we assume that SU can choose to hide the actual eavesdropper channel gain $b$ and PU has only statistical information about $b$. Hence we have a non-cooperative static game with {\it incomplete information} or a {\it Bayesian game} \cite{Myerson:Game91}. This model is more practical in the sense that SU would not be willing to reveal this private information in general.

As introduced in Section \ref{sec:background}, the objective of the players in the Bayesian game is to maximize the expected value of their utility payoff functions. Here, PU has a single type while the type of SU is the actual value of the eavesdropper channel gain $b$. The BE depends on the probability distributions that represent the belief of the players of the unknown parameters. In this section, we assume that PU believes that the eavesdropper channel gain is a realization of a Rayleigh channel. Hence, $b$ is exponentially distributed with mean $\bar{b}$, i.e., $p_b(b')=\frac{e^{\frac{-b'}{\bar{b}}}}{\bar{b}}$. We assume that the average value $\bar{b}$ is known to PU.

The objective of PU in this case is to select a power level to maximize the expected utility payoff function with respect to $b$, which is given by
\begin{align}
{\mathbb E}_b &[u_0(P_0,\alpha)]\notag\\
=& \int_0^\infty{ \left( C(aP_0)-(1-\alpha)C(bP_0)-\gamma P_0\right) \frac{1}{\bar{b}}} e^{\frac{-b}{\bar{b}}}db\notag\\
=& C(aP_0)-\frac{(1-\alpha)}{\bar{b}} \int_0^\infty {C(bP_0)\ e^{\frac{-b}{\bar{b}}}db} - \gamma P_0\notag\\
=& C(aP_0)-\frac{(1-\alpha)}{2\ln(2)} e^{\frac{1}{\bar{b} P_0}} \Gamma(0,\frac{1}{\bar{b} P_0}) - \gamma P_0\label{eq:u1new}
\end{align}
where $\Gamma(s,x)$ is the upper incomplete Gamma function and the last equality follows by \cite{Wolfram:Mathematica7} for $P_0\neq 0$. It is known that for real and positive $x$, we have
\begin{align}
\Gamma(0,x)= - \Ei(-x) = \E_1(x)
\end{align}
where $\Ei(x)=\int_{-\infty}^x{\frac{e^{t}}{t}dt}$ is the exponential integral function and $E_1(x)=\int_{x}^{\infty}{\frac{e^{-t}}{t}dt}$. The eavesdropping term in \eqref{eq:u1new} is concave and increasing in $P_0$ which has similar structure to the eavesdropping term in the case of known eavesdropper channel in previous sections. Hence, it can be seen that ${\mathbb E}_b[u_0(\cdot)]$ and the equilibrium of the game in this section will have the same structure as their corresponding parts in Sections \ref{sec:nash}, \ref{sec:stackelberg}. However, it is hard to maximize \eqref{eq:u1new} analytically. Consequently, in the rest of this section, we perform a numerical study to compute and compare equilibria of the new game to the case when the eavesdropper channel is known at the PU. 

Let the optimal response of the PU in the hidden $b$ case be $P_b(\alpha)$, i.e., $P_b(\alpha)$ maximizes ${\mathbb E}_b [u_0(P_0,\alpha)]$ for some given $\alpha$ and $\bar{b}$. Note that the function $P_b(\alpha)$ is fixed for a given $\bar{b}$ while the function $P^*(\alpha)$ depends on the actual realization of $b$. In Figure \ref{fig:P1alpha}, we compare $P_b(\alpha)$ to $P^*(\alpha)$ for a realization of the eavesdropper channel gain $b=0.7$ where we also set $a = 3, \bar{b} = 0.7, \beta = 1, \bar{\gamma} = 1,  P^{\text{max}}_0 =  P^{\text{max}}_1  = 10$. It is clear that both curves meet at $\alpha=1$ since the eavesdropping term in $u_0(\cdot)$ vanishes at this point where ${\mathbb E}_b [u_0(P_0,1)] = u_0(P_0,1)$.
For $\alpha<1$ and a given eavesdropper channel realization $b$, PU uses a higher power level if it only knows $\bar{b}$ than that used when the actual realization of $b$ is known, i.e., $P_b(\alpha)>P^*(\alpha)$ in this case. In general, $P_b(\alpha)\neq P^*(\alpha)$ and hence PU's expected utility decreases in the hidden $b$ scenario with respect to the known $b$ scenario.
\begin{figure}[thb]
	\centering
    \includegraphics[width=0.8\columnwidth]{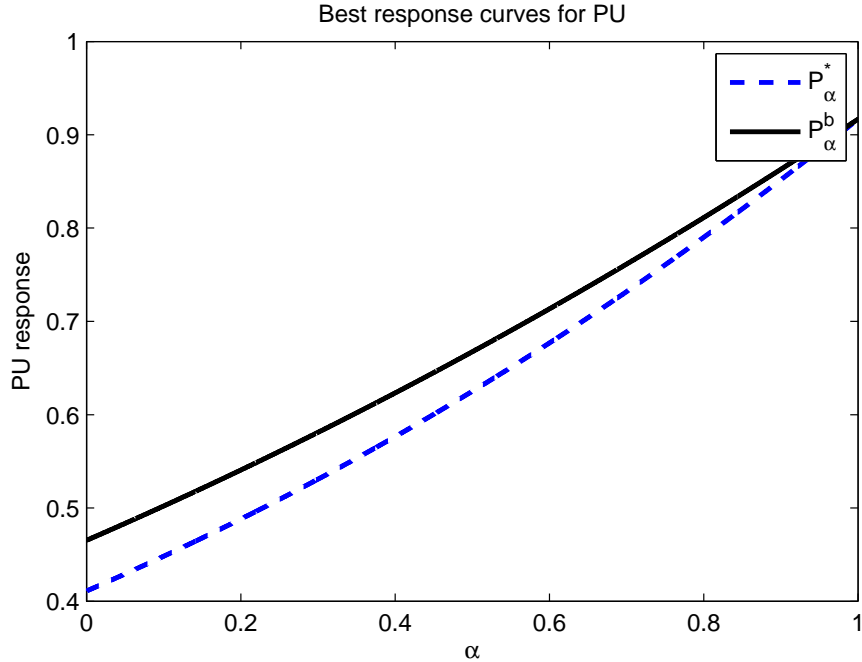}
    \caption{Example primary rational reaction curves for the channel realization $b=0.7$ and $a = 3, \bar{b} = 0.7, \beta = 1, \bar{\gamma} = 1,  P^{\text{max}}_0 =  P^{\text{max}}_1  = 10$}
    \label{fig:P1alpha}
\end{figure}

Next, we compare the utility achieved by PU and SU for different values of the mean $\bar{b}$ and the secondary channel gain $c$. At the SEP of $\Gc$, PU specifies the amount of information allowed for SU, i.e., the value of $\epsilon$. In this section, we fix $\epsilon=10^{-2}$. In Figure \ref{fig:u1fig}, the utility of PU is plotted and compared for both cases of knowledge of $b$ at the PU, when $c=0.7$ and $c = 1.3$. We compare average performance where the utility is averaged over 10,000 realizations of $b$. In this comparison, we use $ P^{\text{max}}_0  =  P^{\text{max}}_1  = 5$, $a = 3$, and $\beta=\bar{\gamma} =1$. It is clear that knowing the exact value of $b$ implies utility that is never less than the other case, i.e., knowledge of $b$ can not hurt PU. In addition, the utility in the hidden $b$ scenario approaches that of the known $b$ scenario for larger range of $\bar{b}$ when $c$ is large. This result can be expected from Lemma \ref{lemma:sep1} and Lemma \ref{lemma:sep2}, where the effect of eavesdropping (and hence not knowing the actual value of $b$) vanishes when $c$ is large enough so that $Q > P^*(1)$ with high probability.
\begin{figure}[thb]
    \centering
    \includegraphics[width=0.75\columnwidth]{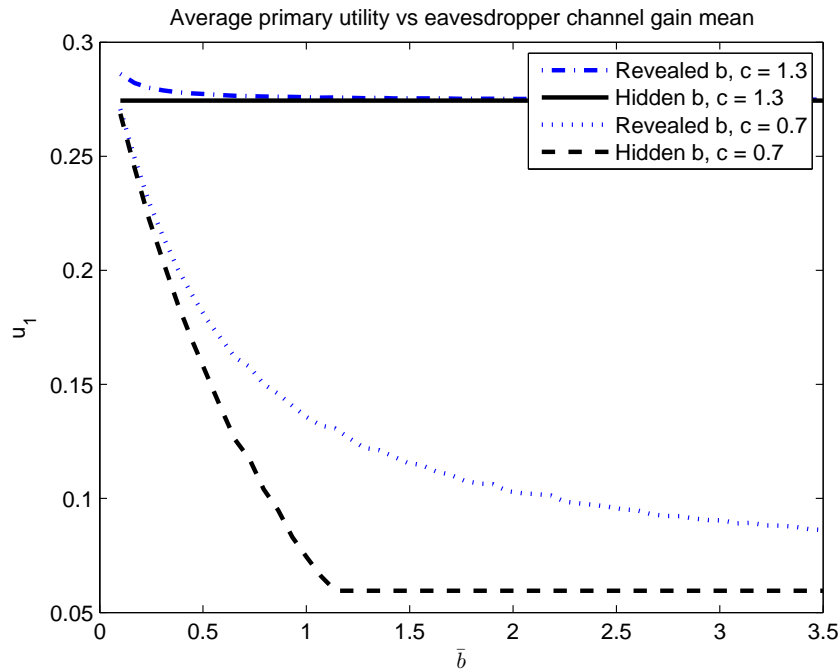}
    \caption{Utility of PU for revealed and hidden $b$ cases.}
    \label{fig:u1fig}
\end{figure}

Finally, we study the effect of hiding $b$ on the utility of SU. When $c\leq 0.6$, $Q \leq 0$ and SU chooses $\alpha=0$ for all $P_0$. In this case, $u_1(\cdot)=0$ for either case of knowledge of $b$ at the PU for all the considered range of $\bar{b}$. In Figure \ref{fig:u2fig}, we plot $u_1(\cdot)$ versus $\bar{b}$ using the same parameters as in the previous case. Here, it is clear that not revealing the actual value of $b$ improves the utility of SU by confusing PU and lowering down its utility, for large values of $\bar{b}$. However, for small values of $\bar{b}$, SU can in fact decrease its own utility by hiding the actual value of $b$. 
\begin{figure}[thb]
    \centering
    \includegraphics[width=0.75\columnwidth]{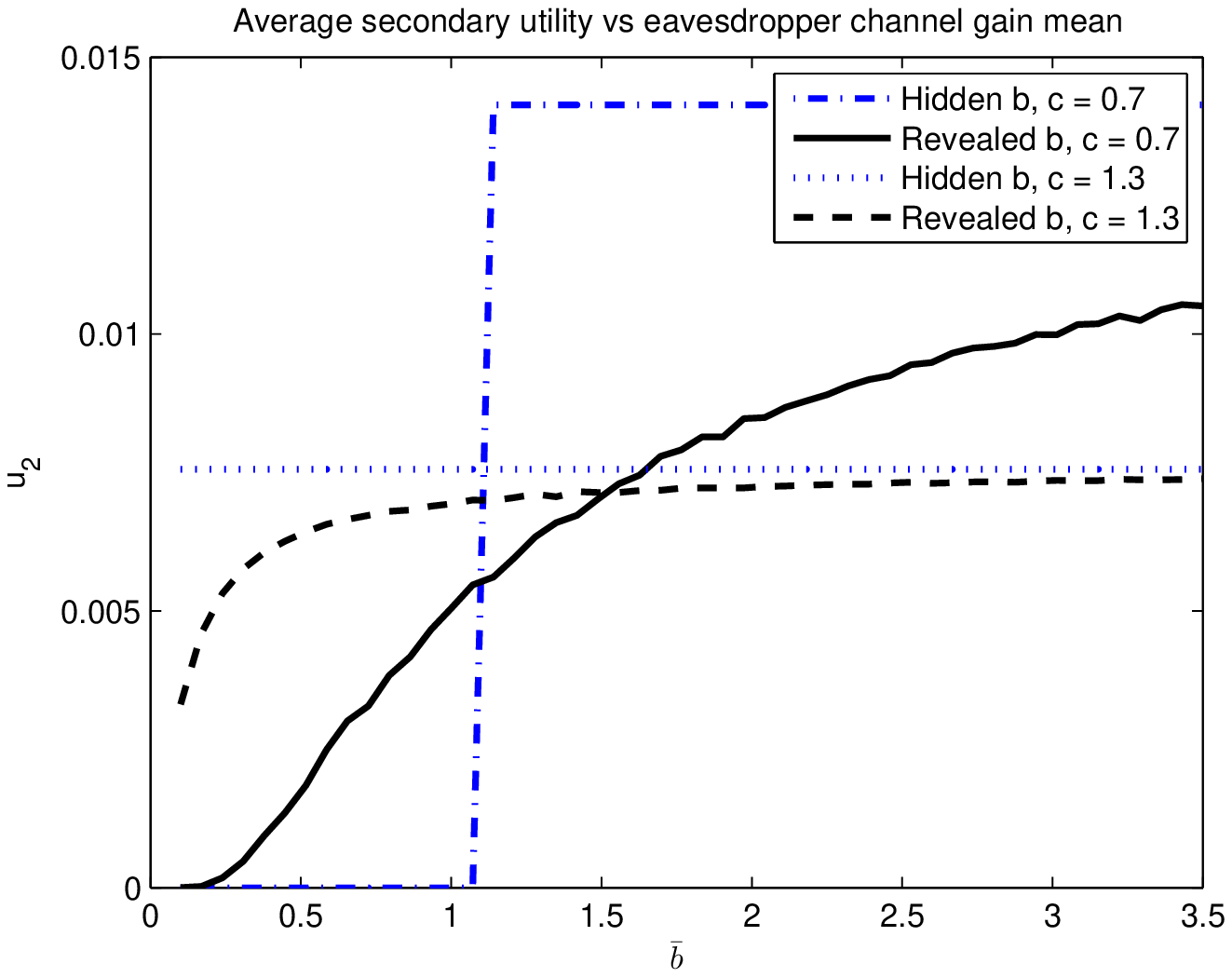}
    \caption{Utility of SU for revealed and hidden $b$ cases.}
    \label{fig:u2fig}
\end{figure}

This last result can be attributed to the following reason. First, when $\bar{b}$ is small so that actual value of $b$ is small with high probability, $u_0(\cdot)$ and ${\mathbb E}[u_0(\cdot)]$ are concave in $P_0$. Then, the decision of PU is made according to \eqref{eq:s1sep1} if SU reveals $b$ and similarly if $b$ is hidden. In addition, when $c$ is small so that $Q<P'(0)$ and $Q<P_b(0)$, we have a range of $\bar{b}$ values such that
\begin{align}
u_0(P'(0),0)<u_0(Q-\epsilon,1)<{\mathbb E}[u_0(P_b(0),0)],
\end{align}
where in this range, the optimal decision for PU is $P_0 = Q-\epsilon$ in case of revealed $b$ and $P_0=P_b(0)$ for hidden $b$ case. Hence, SU's decision (as a follower) is $\alpha = 1$ and $\alpha = 0$ in the cases of revealed and hidden $b$, respectively. Thus, {\it SU is willing to reveal the actual value of $b$ to achieve better utility in this case.}


\section{Multiple-User Game}
\label{sec:multi}
In this section, we extend the two-player model considered in the previous sections to the case when there exist multiple primary and secondary users in the network communicating to a common destination. Similar to previous sections, we consider successive interference cancellation at the decoder. In this multiple secondary users setting, the decoder is required to find the optimal order for decoding different users such that the utility of the PU is maximized at the equilibrium. First, we extend the results of Sections \ref{sec:nash} and \ref{sec:stackelberg} and show that a similar conclusion holds for the extended game given a general decoding order. Then, we show through an example that finding the optimal decoding order is a combinatorial problem. Next, we focus on a special case for the parameters of the SUs and show that the optimal decoding order can be analytically characterized. Specifically, we show that it is optimal for the primary system to give higher decoding priority to stronger eavesdroppers in this case. Finally, based on the analysis, we present and numerically evaluate an algorithm for the primary system to select SUs for decoding and find the optimal power level for the PU.

Here, we focus on a network model where resources are assigned to primary users orthogonally. That is, each PU is operating exclusively on a separate time-frequency resource. This model covers a wide range of practical wireless networks, including LTE networks where each user is separately assigned a number of Physical Resource Blocks (PRBs). Recall that we consider one shot games that apply for a single time slot. Thus, without loss of generality, we will consider the case where the network is composed of one PU and multiple SUs all communicating to a common destination on a single physical resource. We also assume that SUs eavesdrop the transmission of the PU independently, i.e., we do not consider colluding SUs. We consider the case when all eavesdropper channel gains are known at the base station similar to Sections \ref{sec:nash} and \ref{sec:stackelberg}.

When multiple eavesdropping SUs are present, the utility of the PU is given by
\begin{align}
u_0(P_0,\alpha_1,\alpha_2,\cdots,\alpha_N) = C(aP_0)-\max_{i\in\Nc}\left\{(1-\alpha_i)C(b_iP_0)\right\} - \gamma P_0
\label{eq:u0all}
\end{align}
where $N$ is the number of SUs, $\Nc = \{1,2,\cdots,N\}$; $b_i$ is the eavesdropper channel gain of SU $i$. For some given $P_0$, we say that SU $k$ with $k = \argmax_i {(1-\alpha_i)C(b_i P_0)}$ dominates the utility of the PU. We also use an approximation to PU's secure rate based on Lemma \ref{lemma:approx}. The utility function of SU $i$ is given by
\begin{align}
u_i(P_0,\alpha_1,\alpha_2,\cdots,\alpha_N) = \alpha_i\left[ C\left(\frac{c_i P_i^{\text{max}}}{1 + a P_0 +\sum_{j\in\Nc'(i)}{\alpha_j c_jP_j^{\text{max}}}}\right)-\beta_i\right]
\label{eq:ui}
\end{align}
where $c_i$ and  $P_i^{\text{max}}$ are the secondary channel gain and the transmission power of SU $i$; $\Nc'(i)\subset \Nc$ is the set of SUs that are given higher priority than SU $i$ at the decoder (i.e., signal of SU $j\in \Nc'(i)$ is decoded after decoding signal of SU $i$), and $\beta_i$ is the energy cost parameter for SU $i$. Suppose, without loss of generality, that $b_1\geq b_2\geq \cdots \geq b_N$. We define the primary utility $u^i_0(P_0,\alpha_i)$ when only SU $i$ is present as follows.
\begin{align}
u^i_0(P_0,\alpha_i) = C(aP_0)-(1-\alpha_i)C(b_iP_0) - \gamma P_0,
\end{align}
$\forall i\in\Nc$. We also define $u^0_0(P_0) = u_0(P_0,1,\cdots,1) = C(aP_0) - \gamma P_0$. It can be seen that $u_0(P_0,\alpha_1=1,\alpha_2=1,\cdots, \alpha_{i-1}=1,\alpha_{i},\cdots,\alpha_N)=u_0^i(\cdot)$. We define PU's power level that maximizes $u_0^i(\cdot)$ given $\alpha_i$ as $P^i(\alpha_i)$. Note that $P^i(1)=P^*(1)$ for all $i$, where $P^*(1)$ is given by \eqref{eq:P11}. For an SU $i$, we define the threshold
\begin{align}
Q_i (\alpha_1,\alpha_2,\cdots,\alpha_{N}) = \frac{1}{a}\left(\frac{c_i  P^{\text{max}}_i }{2^{2 \beta_i}-1}-1-\sum_{j\in\Nc'(i)}{\alpha_j c_jP_j^{\text{max}}}\right).
\label{eq:Qi}
\end{align}

\subsection{Equilibrium Analysis}
\label{sec:eq_anal}
In this section, we characterize the equilibria of the multi-player game where the PU and SUs are the players. We note that characterization of Nash equilibria of the mutli-player game is challenging, where the difficulty is mainly due to the structure of the utility of the PU. 
In the following, we show that the results we presented for the two-player game in Sections \ref{sec:nash} and \ref{sec:stackelberg} also hold for the multi-player game. Specifically, we first describe the structure of the NE for any given decoding order. Then, based on this description, we show that the SEP dominates any NE of the multi-player game while any SES (i.e., any Stackelberg equilibrium where one SU is the leader and PU as well as other SUs are followers) can not dominate the NE.

For the multi-player game, it can be seen that the game is a continuous game, i.e., the utility functions are continuous and the strategy sets are compact. Thus, we know that there exist an NE for this game \cite{Basar:Dynamic95}. The utility functions of SUs  are concave in their corresponding strategies. Thus, SUs will employ pure strategies at the equilibrium. 
We also note that at NE, since the utility of each SU is linear in the corresponding strategy, then if $\alpha_i^* < 1$, then $u_i^* = 0$.

Similar to the two-player game in Section \ref{sec:stackelberg}, we show that the outcome of the game will be the equilibrium of the leader follower game with the PU as leader and SUs as followers. 

\begin{lemma}
For the multi-player game with a given decoding order, SEP dominates NE.
\label{lemma:sep-multi}
\end{lemma}
\begin{proof}
We show that at the SEP, the utility of any follower in the game (i.e., any SU) cannot be less than that at an NE. Fix some decoding order. Suppose the strategies of the PU and SU $i$ at SEP are $P_0^{SEP},\alpha_i^{SEP}$. Note that PU cannot improve its performance with respect to NE by increasing its power level, i.e., $P_0^{SEP}\leq P^*_0$. Thus, at the SEP, the utility of any SU cannot be reduced compared to that at NE. Specifically, if $P^*_0 \geq Q_i$ and $P_0^{SEP} < Q_i$, then $\alpha_i^* < 1$ and $\alpha_i^{SEP} = 1$. In this case, the utility of SU $i$ is improved. In addition, for every SU $j$ with a lower decoding priority such that $Q_j = Q_j(\alpha_i)\leq Q_i$, $Q_j$ is reduced with respect to NE. However, since any SU $j$ with $\alpha_j^* < 1$ achieves zero utility at NE, then this decrease in $Q_j$ will not decrease $u_j$, concluding the proof.
\end{proof}

\begin{lemma}
For the multi-player game with a given decoding order, there exist no SES that dominates NE.
\label{lemma:ses-multi}
\end{lemma}
\begin{proof}
We compare the utility of different players achieved at NE to that achieved at an SES when SU $i$ is the leader, and show that SES cannot dominate NE. Fix some decoding order. Suppose $\alpha_i^* = 1$. Then if $\alpha_i^{SES} = 1$, then same utility is achieved by all other players (followers). In addition, SU $i$ can not improve its utility by selecting $\alpha_i^{SES} < 1$ since this can only increase the threshold values for all lower priority SUs which implies increased $P_0$ at the equilibrium and thus decreased $u_i$. Now suppose $\alpha_i^* < 1$. Thus, the utility of SU $i$ at NE is $u_i^* = 0$. If $\exists j$ such that $b_j > b_i$ and $Q_j\geq Q_i$, then SU $i$ can not improve its utility and SES is the same as NE. Note that if $\exists j$ such that $b_j > b_i$ and $Q_j < Q_i$, then $\alpha_i^*=1$. Finally, if SU $i$ dominates $u_0$ at NE, then selecting $\alpha_i^{SES}<\alpha_i^*$ will either not improve $u_i$ (if $Q_i<P^i(0)$) or will decrease the utility of the PU, compared to that achieved at NE. If $\alpha_i^{SES}=1$, then $Q_j$ will be decreased for any $Q_j\leq Q_i$ and $u_0$ will be decreased. Thus, for any possible NE, SES cannot dominate NE. 
\end{proof}

The following proposition concludes that the main result for the two-player game is also extendible to the multi-player game for a given decoding order. The proof follows immediately from Lemma \ref{lemma:sep-multi} and Lemma \ref{lemma:ses-multi}.
\begin{proposition}
Fix a given decoding order, the equilibrium of the multi-player game will be the $\epsilon$-SEP. 
\end{proposition}


The secure rate of the PU is characterized by the order of the eavesdropper channel gains as well as the strategies of SUs. Let $j = \argmax_{i\in\Nc} b_i$. In general, if SU $j$ chooses to transmit by selecting $\alpha_j=1$, then the utility of PU will be improved and will depend on the channel of the second strongest eavesdropper channel. In addition, the primary system can decode transmissions from multiple SUs if this will improve the utility achieved by the PU. Note that an SU starts transmitting its own information and stops eavesdropping primary traffic only if its achieved rate is larger than its energy cost. Thus, in general, not all SUs are willing to switch from eavesdropping to transmission. In the most general case where multiple SUs are allowed to coordinate their coding schemes with the PU and successive interference cancellation is employed at the destination, the order of decoding will affect the decision of each SU since signals decoded first will suffer interference from all other signals to be decoded next (PU and other transmitting SUs) \cite{Cover:Elements91}. This can be seen from the structure of the utility function \eqref{eq:ui}.

It is clear that since the decoder always gives priority to the PU, then decoding PU's signal is always last in order so that interference from signals of other transmitters (i.e., SUs) do not affect the first term in \eqref{eq:u0all}. However, finding the decoding order for SUs that maximizes the utility of the PU is not straight forward. In fact, the decoding order will affect the equilibrium strategies of the multiple player game comprising the PU and all SUs as players. Specifically, the decoding order will affect the threshold on the PU transmission power level $Q_i$ associated with each SU $i$ above which SU $i$ is not willing to transmit its own information and will choose to always eavesdrop the transmission of the PU. In Section \ref{sec:stackelberg}, it was shown that the equilibrium of the game is determined according to the threshold value of the PU. Specifically, it was shown that if $Q > P^*(1)$ then both players achieve their maximum possible utility. However, in the other cases, the SU achieves zero utility if $Q$ is small and achieves positive utility if $Q$ is large enough (Lemmas \ref{lemma:sep1} and \ref{lemma:sep2}). As mentioned above, the utility of the PU is function only in the largest eavesdropped rate. Thus, from primary system's perspective, it is plausible to give higher priority in the decoding order to the SUs with higher eavesdropping capability (larger $b_i$) such that the threshold values of those stronger eavesdroppers are the largest and PU can achieve higher utility at the equilibrium of the game. However, in general, giving higher decoding priority to stronger eavesdroppers might not necessarily imply a better PU utility at the equilibrium of the game. The example in Figure \ref{fig:example2} illustrates this fact.


\begin{figure}[thb]
    \centering
    \includegraphics[width=0.85\columnwidth]{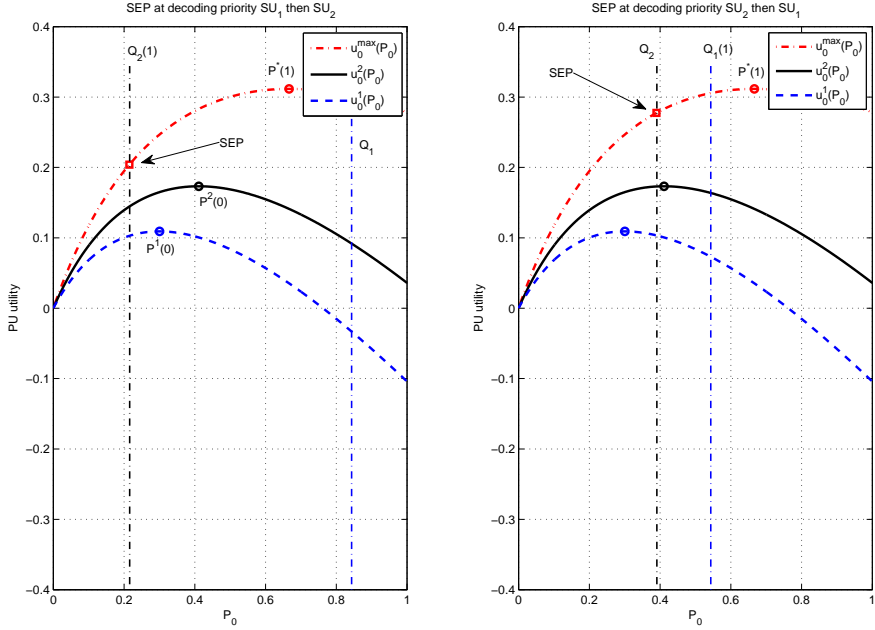}
    \caption{Example illustrating the optimal decoding order at equilibrium when $N$=2. Here, $a=3,\gamma=\log2(4)^{-1},b_1=0.7,b_2=0.4,c_1 = 0.6, c_2 = 0.35, \beta_1 = 0.1, \beta_2 = 0.25, P_1^{\text{max}}= P_2^{\text{max}} = 1.5$. When SU 2 has priority in decoding, the value of $u_0$ at the equilibrium is larger than its value when SU $1$ has the decoding priority. Thus, in this example, it is optimal to give SU $2$ decoding priority.}
    \label{fig:example2}
\end{figure}

The above example shows that, in general, it is not possible to characterize the optimal decoding order at the base station given only the order of the eavesdropping channel gains. Thus, it is not sufficient to consider how the SUs are ordered in terms of their eavesdropping channel gains to find the optimal decoding order that maximizes the PU utility at the equilibrium. In other words, the problem becomes of combinatorial nature and requires the comparison between the equilibrium outcome of the game for ($N!$) different possible orderings of decoding. In the next section, we focus on studying the equilibria of the multiple SU game for the special case when secondary user parameters are {\it uniform}. In this case, the decoding order necessarily affects the order of the thresholds of SUs. For this case, we show that the optimal order is to always give higher priority to the stronger eavesdropper, i.e., to decode signals in ascending order of $b_i$. 

\subsection{The Case of Uniform Secondary User Parameters}
\label{sec:uniform}

In this section, we consider a case where the secondary channels and cost parameters of secondary users are such that $c_i P_i^{\text{max}} = k_1$ and $\beta_i = k_2$, $\forall i\in\Nc$ for some constants $k_1,k_2>0$. We call this case the {\it uniform secondary user parameter	} case.  

First, note that when the PU announces its strategy in the Stackelberg game, all SUs will then play a Nash game and the Nash equilibrium (in response to the strategy announcement by the PU) will determine the equilibrium strategy profile of the Stackelberg game. Let the SU with highest decoding priority be $j$. Then, $\alpha_j^{SEP}$ is only function in $P_0^{SEP}$ since $Q_j$ is a constant. Suppose the SU with second highest decoding priority is $k$. In this case, $\alpha_k^{SEP}$ is a function in $\alpha_j^{SEP}$. Following along this direction, the NE of the followers game can be specified given the strategy of the leader. Suppose, without loss of generality, that $b_1\geq b_2 \geq \cdots \geq b_N$. Now if PU selects $P_0^{SEP}$ such that at the equilibrium we have $P_0^{SEP}> Q_1$, then we have $u_0 = u^1_0(P_0^{SEP},0)$ and the strategies of other SUs are irrelevant to PU's utility. In this case, $P_0^{SEP} = P^1(0)$. As shown in the two-player game in Section \ref{sec:stackelberg}, this choice may not be optimal if decreasing $P_0^{SEP}$ increases $u_0$ where $\alpha_1^{SEP}=1$ in this case. In the multiplayer game, the PU utility achieved if $P_0^{SEP} = Q_1-\epsilon$ is $u_0^l(P_0^{SEP},0)$, such that $l=\argmax_{i\in \Sc(P_0^{SEP})} b_i$ and $\Sc(P_0^{SEP})$ is the set of SUs with $Q_i<P_0^{SEP}$ at equilibrium. Every SU $i$ with $i\notin \Sc(P_0^{SEP})$ will have $\alpha_i^{SEP}=1$. We are now ready to present the result about the optimal decoding order.

\begin{proposition}
The decoding order $N,N-1,\cdots,1$ maximizes $u_0$ at the equilibrium of the multi-player game.
\end{proposition}
\begin{proof}
It can be seen that $P_0^{SEP}=P^*(1)$ if and only if $Q_i(1,1,\cdots,1)>P^*(1)$ for all $i\in\{1,2,\cdots,N\}$. If $Q_i<P^*(1)$ for any $i$, then $P_0^{SEP}<P^*(1)$. Only in this case $P_0^{SEP}$ is function in the decoding order (i.e., the order of $Q_1(\cdot),Q_2(\cdot),\cdots,Q_N(\cdot)$).

Now suppose $Q_i<P^*(1)$ for some $i$. To find $P_0^{SEP}$, we first recall that $u_0$ is dominated by SU $i$ with largest $b_i$ within the set of SUs with $\alpha_i=0$. When $P_0 <Q_i$, $u_0$ is dominated by the SU with next largest eavesdropper channel. Thus, we need to check the value of $u_0(P_0,\alpha_1,\cdots,\alpha_N)$ for all possible values for $P_0$ within the set $\{Q_i(1,1,\cdots,1)-\epsilon,P^i(0)\}$ for all $i$ such that $Q_i(1,1,\cdots,1)<P^*(1)$ and in the order of decreasing $b_i$. Let the search space be denoted as $\Pc$. It can be seen that both $|\Pc|$ and the value of $P_0^{SEP}$ that maximizes $u_0(\cdot)$ will depend on the order of $Q_i$ and thus on the decoding order. Specifically, first, choosing a decoding order such that $Q_i\leq Q_j$ while $b_i\geq b_j$ can only decrease $|\Pc|$. This follows from the fact that $Q_j(1,1,\cdots,1)-\epsilon \notin |\Pc|$ if there exists $i$ such that $Q_i(\cdot)\leq Q_j(\cdot)$ and $b_i\geq b_j$ since either $P_0^{SEP}<Q_i(\cdot)$ or $P_0^{SEP}=P^i(0)$ where in both cases SU $j$ cannot dominate $u_0(\cdot)$ and thus $Q_j(1,1,\cdots,1)-\epsilon \notin |\Pc|$. Second, from the structure of $u_0(\cdot)$, it can be seen that $u^1_0(P_0,0)\leq u^2_0(P_0,0)\leq \cdots \leq u^N_0(P_0,0) \forall P_0$. In addition, since $u^i_0(P_0,0)$ is increasing for $P_0\leq P^i(0)$, then $P^i(0)\geq P^j(0) \forall\ i,j$ with $b_i\geq b_j$ (i.e., $P^i(0)$ is non-decreasing in $b_i$). Hence, for any $Q_i(\cdot)-\epsilon \in |\Pc|$, $u_0$ is the largest at $Q_i(\cdot)-\epsilon$ if the decoding order is such that $Q_1\geq Q_2\geq\cdots\geq Q_N$.
\end{proof}
In Figure \ref{fig:exampleProof}, we present an example to illustrate the idea of the proof of the decoding order that maximizes the utility of the primary user, where $N=2$.
\begin{figure}[thb]
    \centering
    \includegraphics[width=0.85\columnwidth]{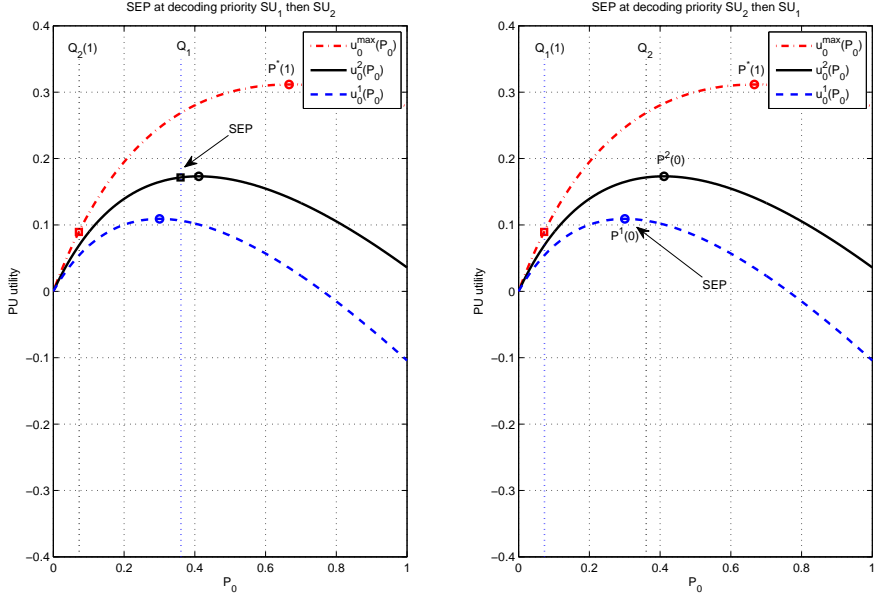}
    \caption{Example for the illustration of SEP analysis when $N$=2. Here, $a=3,\gamma=\log2(4)^{-1},b_1=0.7,b_2=0.4,k_1=0.8625,k_2=0.25$. In this example, when SU 2 is given priority in decoding, $Q_1(1)$ is very small and PU cannot increase its utility above $u_0^1(P^1(0),0,\times)$. Here, $\alpha_1^{SEP} = 0$ and the the value of $\alpha_2^{SEP}$ is irrelevant to the PU since SU 1 dominates $u_0(\cdot)$. Thus, the base station can simply block access to SU 2 and thus $\alpha_2^{SEP} = 0$ as well.}
    \label{fig:exampleProof}
\end{figure}

In the following section, we propose an iterative algorithm that can be implemented at the base station to find the value of $P_0^{SEP}$ at the equilibrium as well as find the set of SUs whose signals are to be decoded to maximize the utility of the primary user. 

\subsection{Proposed Algorithm}
\label{sec:alg}
Based on the analysis in the previous sections, we propose an algorithm that iteratively calculates the value of $P_0$ at the equilibrium of the game, for a fixed decoding order. The algorithm also constructs the subset of all SUs in the network to decode their signals. 

Given feedback from the base station, the PU will be able to calculate its transmission power level that maximizes the utility $u_0$. Since calculation of the optimal primary power level requires some processing, we suggest that this processing is done at the base station and then the optimal power level is sent to the PU to use. Given the channel state information of the network, the base station orders the list of all SUs in an ascending order according to their eavesdropper channel gains. For the case of uniform secondary user parameters, the utility of the PU is the maximum possible at this decoding order. Then the base station announces the identities and order of SUs that will have their signals decoded. Given this announcement, the PU will start transmission at the power level that maximizes it's utility given the selection of SUs. This power level is actually the Stackelberg equilibrium strategy for the leader derived in Section \ref{sec:eq_anal}. In addition, SUs that are not selected will immediately start eavesdropping (choose $\alpha_i = 0$) while the selected SUs will transmit their own information (choose $\alpha_i = 1$). 

For a game with only PU and SU $i$, it can be seen from Lemmas \ref{lemma:sep1}, \ref{lemma:sep2} that the Stackelberg equilibrium strategy for the PU is given by:
\begin{align}
P^{SEP} =
\begin{cases}
\argmax_{P\in\{Q_i-\epsilon,P^i(0)\}}u^i_0(\cdot) &\text{if }Q_i \leq P'(1)\\
P^i(1); &\text{if } P^i(1) < Q_i.
\label{eq:P}
\end{cases}
\end{align}

The following iterative algorithm defines the selected SU set $Allowed \ SUs$ and PU power level $P$ that maximizes the utility function of the PU in \eqref{eq:u0all}. At iteration $i$, the strongest SU channel, within SUs that are not selected, is considered, the optimal PU power value is updated and the decision to add another SU to the set $Allowed \ SUs$ is made based on \eqref{eq:Qi} and \eqref{eq:P}. The algorithm considers the case where $N>2$. It also only considers cases where there exists at least one SU $i$ with $Q_i(1,1,\cdots,1)<P^*(1)$ since it is easy to see that $P_0^{SEP} = P^*(1)$ and the all SUs are granted spectrum access in the other case.

\begin{algorithm}
	\caption{Algorithm for finding SUs to grant spectrum and calculating $P_0^{SEP}$}
	\begin{algorithmic}
	\State begin initialization
	\State $Allowed\ SUs=\{k:Q_k(1,1,\cdots,1)\geq P^*(1) \}$, $j = \argmax_{i\in\Nc,Q_i(\cdot)<P^*(1)}b_i$, $P_0^{SEP} = 0$, $u_0^{SEP} = 0$
	\State end initialization
	\For {$i=j$ to $N$}
		\If {$Q_i(1,1,\cdots,1) \geq P^i(0)$}
			\State $P_0^{SEP} \gets Q_i(1,1,\cdots,1)-\epsilon$
			\State $u_0^{SEP} \gets u_0^{i+1}(Q_i(1,1,\cdots,1)-\epsilon,0)$
			\State $Allowed\ SUs \gets Allowed\ SUs\cup \{i\}$
		\Else
			\If {$u_0^i(P^i(0),0) \geq u_0^{SEP} \&\ P^i(0)\leq Q_{i-1}(1,1,\cdots,1) $}
				\State $P_0^{SEP} \gets P^i(0)$
				\State $u_0^{SEP} \gets u_0^i(P^i(0),0)$
			\EndIf
			\If {$u_0^{i+1}(Q_i(1,1,\cdots,1)-\epsilon,0) \geq u_0^{SEP}$}
				\State $P_0^{SEP} \gets Q_i(1,1,\cdots,1)-\epsilon$
				\State $u_0^{SEP} \gets u_0^{i+1}(Q_i(1,1,\cdots,1)-\epsilon,0)$
				\State $Allowed\ SUs \gets Allowed\ SUs\cup \{l:l\leq i\}$
			\EndIf
		\EndIf
	\EndFor
	\end{algorithmic}
\end{algorithm}

{\it Description of the algorithm:} The algorithm starts by adding all SUs with large threshold (larger than $P^*(1)$) to the set of allowed SUs. In addition, the search for $P_0^{SEP}$ is started by considering the parameters of the SU with strongest eavesdropper channel and with threshold less than $P^*(1)$, since this SU is dominating $u_0(\cdot)$ when it is eavesdropping. The effect of eavesdropping of each SU is then assessed. For each SU $i$, we check (in the order of descending $b_i$) the value of $u_0$ at both $P_0 = Q_i(1,1,\cdots,1)-\epsilon$ and $P_0 = P^i(0)$ and then add this SU to the set of allowed SUs if it improves $u_0$. Note that if $Q_i(1,1,\cdots,1) < P^i(0)$ for some $i$, then $Q_j(1,1,\cdots,1) < P^j(0)$ for all the following $j>i$ since $Q_j(1,1,\cdots,1)<Q_i(1,1,\cdots,1)$ and $P^j(0) > P^i(0)$ for all $i,j$ such that $i<j$. Here, decreasing $P_0^{SEP}$ can either be at $P_0^{SEP} = P^i(0)$ or $P_0^{SEP} =Q_i(1,1,\cdots,1)-\epsilon$. However, $P_0^{SEP} = P^i(0)$ can be an equilibrium point only if $P^i(0)\leq Q_{i-1}(1,1,\cdots,1)$. Moreover, if $P_0^{SEP} = P^i(0)$ at a given iteration, then SU $i$ is not added to the set $Allowed\ SUs$. It is possible that adding some lower priority SU $k>i$ (by choosing $P_0^{SEP}Q_k(1,1,\cdots,1)-\epsilon$ in the following iterations may improve $u_0^{SEP}$. In this case, this SU as well as all the preceding ones with stronger eavesdropper channels are needed to be granted access and thus added to the set $Allowed\ SUs$.

\subsection{Numerical Examples}
\label{sec:perf}
In this section, we present numerical examples showing how the PU utility varies with the number of SUs in the network. In addition, to show the tradeoff between the decoder complexity and the optimality of the primary utility, we compare the average utility achieved by the primary user when no SUs are granted spectrum, when only one SU is granted spectrum access (strongest eavesdropper) and when multiple SUs are granted spectrum according to Algorithm 1. We plot the average primary utility versus the increasing number of SUs. Here, we consider Rayleigh fading eavesdropper channels with mean equals to $1$. We fix $\epsilon=10^{-3}$, $c_i=0.5, P_i^{\text{max}}=4.5\ \forall i$, $P^{\text{max}}_0 = 4.5$, $a = 2$, and $\bar{\gamma} =0.2$. With this choice of parameters, $P^*(1) = 4.5$ and $u_0^0(P^*(1))=1.01$. In Figure \ref{fig:uniform}, we plot the average primary utilities for two cases. In the first case, we set $\beta_i=0.1,\ \forall i$. Here, the values of the thresholds for the $20$ SUs ranges from $7.07$ to $-14.3$ with $Q_i>P^*(1)$ for $i\in\{1,2,3\}$, $Q_i>0$ for $i\in\{1,\cdots,7\}$ and $Q_i<0$ otherwise. In the second case, we set $\beta=0.2$. The values for the thresholds ranges in $[4.3,-16.6]$ with only three positive thresholds. 

In the plots, it can be seen that the gain provided by our algorithm decreases as the number of SUs increase since the probability of having an SU with large eavesdropping channel gain and low threshold increases. For moderate number of SUs, however, the gain from decoding the signals from multiple SUs is apparent, compared to the case where only one SU is granted spectrum. For example, in the second case, by decoding signal of $3$ SUs, on average, the PU utility achieved by our algorithm at $N=10$ is improved more than $10$ times compared to decoding the signal of only one SU. In addition, in the first plot, since SU thresholds are higher, more SUs are granted access on average improving $u_0(\cdot)$ with respect to the second case. In addition, in the first case, we have a flat curve at low $N$ since the first few values for the thresholds are already larger than $P^*(1)$ and thus these SUs are granted access without compromising in the utility of the PU.

\begin{figure}[thb]
    \centering
    \includegraphics[width=0.75\columnwidth]{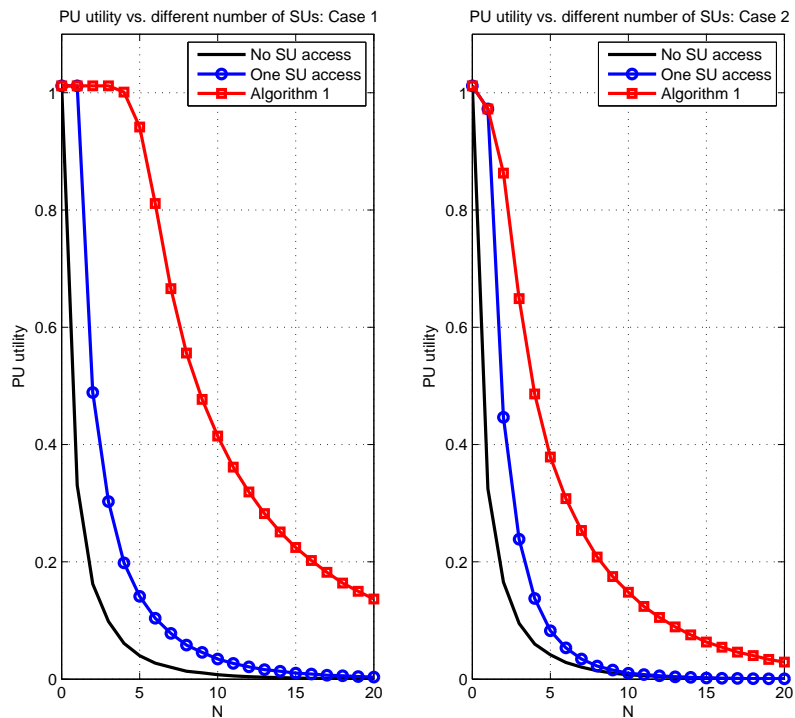}
    \caption{Average PU utility for the case of uniform SUs parameters.}
    \label{fig:uniform}
\end{figure}

In Fig. \ref{fig:avgSUallowed}, we plot the average number of SUs that are granted access. In each of the considered two cases, the number of positive threshold values constitutes a deterministic upper bound, since no SU positive PU power level can turn an eavesdropping SU with negative threshold into transmission mode. By increasing the number of SU, the probability of having smaller values eavesdropping channel gain $b_i$ increases and thus the probability to achieve higher value for $u_0$ at the same (deterministic) threshold level increases. Thus, the average cardinality of the set $Allowed\ SUs$ is increasing as it is evident in the plots.

\begin{figure}[thb]
    \centering
    \includegraphics[width=0.75\columnwidth]{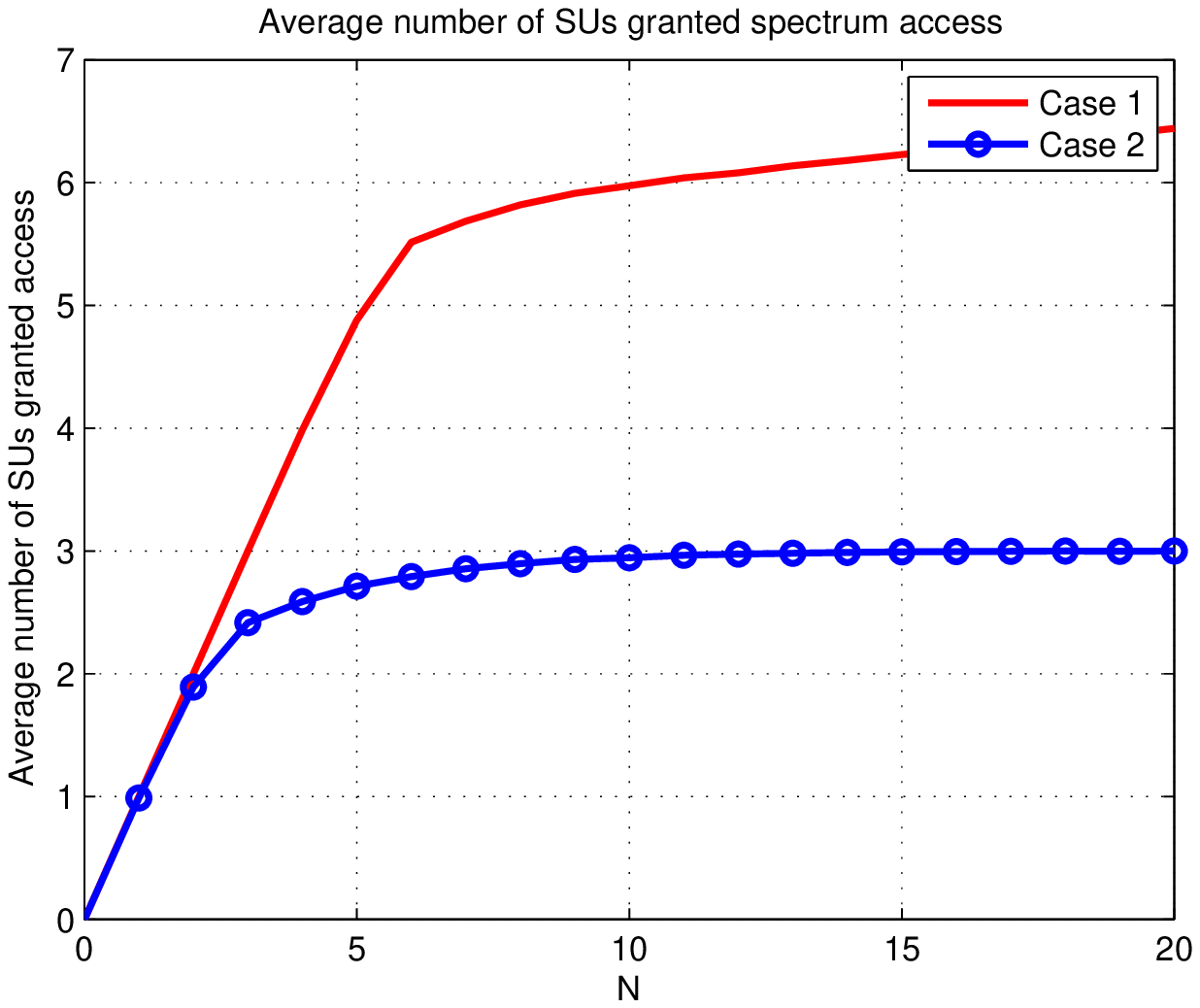}
    \caption{Average number of SUs that are granted spectrum access in case of uniform  SUs parameters.}
    \label{fig:avgSUallowed}
\end{figure}

Finally, in Figure \ref{fig:non-uniform}, we study the average PU utility achieved for the non-uniform case of SU parameters. Here, $c_i$  is chosen randomly according to a Rayleigh distribution with mean equals to $0.5$ and the rest of parameters have the same values as in the previous case. Even though the decoding order employed is suboptimal in this more general setting, the primary system can still achieve performance gain by decoding signals of SUs in this order.

\begin{figure}[thb]
    \centering
    \includegraphics[width=0.75\columnwidth]{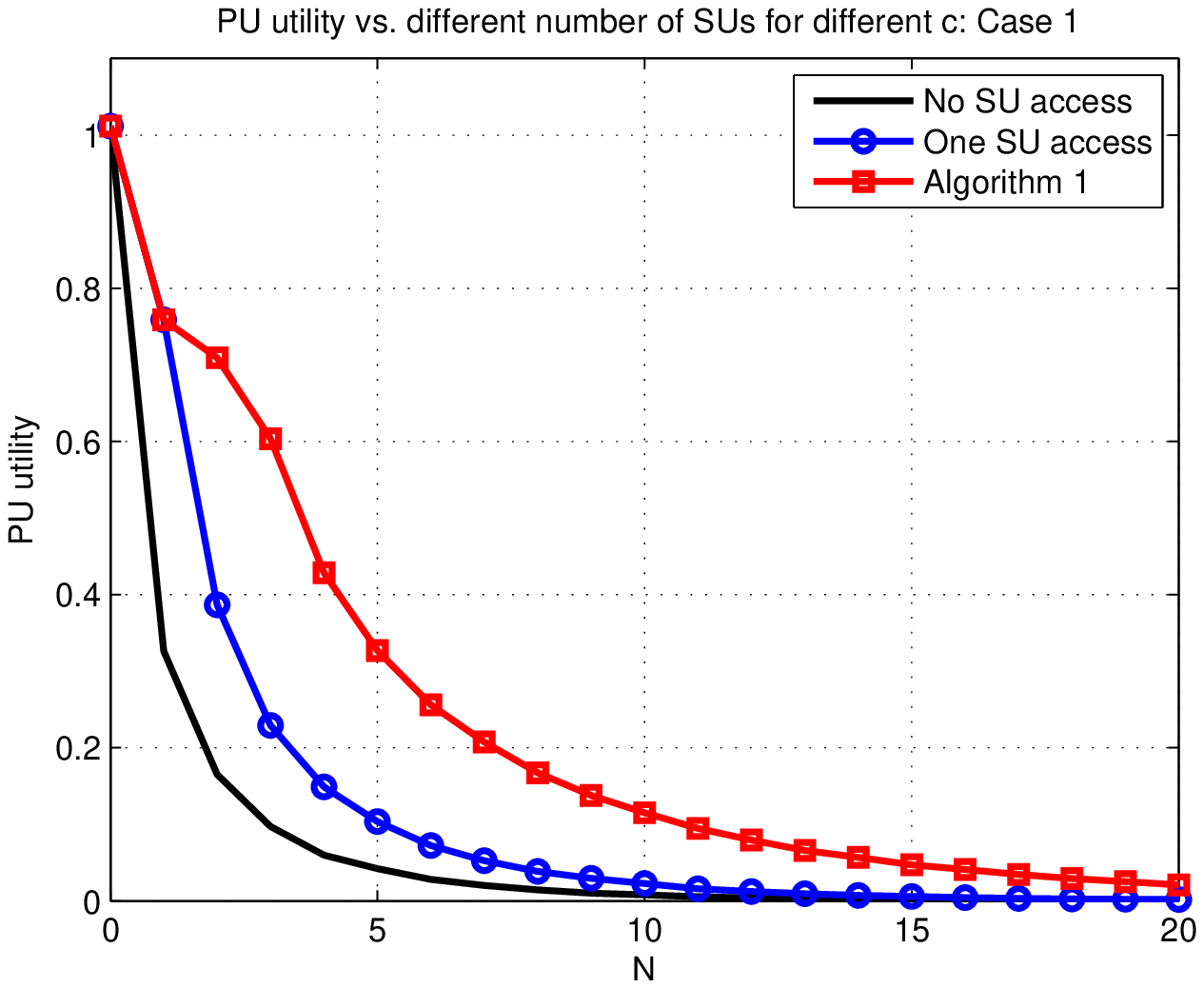}
    \caption{Average PU utility for the case of non-uniform SUs parameters.}
    \label{fig:non-uniform}
\end{figure}

\section{Conclusion}
\label{sec:conclusion}
In this paper, we analyzed an adversarial situation between primary users and secondary users of a CRN communicating with a common destination. Both types of users are interested in maximizing their own data rate at the minimum possible energy. The cognitive half duplex users threaten the primary users to eavesdrop the primary traffic if they are not allowed to access the spectrum and transmit their own information. Using tools from non-cooperative game theory as well as information theoretic transmission strategies, it is shown that the eavesdropping capability of the secondary users may force the primary users to reduce their power level so that the secondary users achieve non-negative utility. However, by using Stackelberg formulation, we show that the primary users can specify the allowable secondary rate and the secondary users are forced to comply to this specification, even if its achieved positive utility is small. Moreover, the result also holds for the stronger case when the primary channel is weaker than the eavesdropper channel. We also presented a numerical study for the case when the eavesdropper channel state is known only statistically at the primary users. By not revealing information about the eavesdropper channel gain, we show that the cognitive users can improve their achieved utility, unless the eavesdropper channel is weak. Finally, based on the equilibrium analysis, we presented an iterative algorithm for the primary system to find the equilibrium strategy that maximizes the performance of the primary users.

\bibliographystyle{IEEEtran}
\bibliography{cognitive-threat-journal}

\balance

\end{document}